\documentclass{article}

\PassOptionsToPackage{numbers,compress,sort}{natbib}
\usepackage[preprint]{neurips_2026}

\usepackage[utf8]{inputenc}
\usepackage[T1]{fontenc}
\usepackage[colorlinks=true,linkcolor=blue!70!black,citecolor=blue!70!black,urlcolor=blue!70!black]{hyperref}
\usepackage{url}
\usepackage{booktabs}
\usepackage{amsfonts}
\usepackage{amsmath}
\usepackage{amssymb}
\usepackage{amsthm}
\usepackage{nicefrac}
\usepackage{microtype}
\usepackage{xcolor}
\usepackage{colortbl}
\usepackage{tabularx}
\usepackage{graphicx}
\usepackage{subcaption}
\usepackage{algorithm}
\usepackage{algorithmic}
\usepackage{multirow}
\usepackage{placeins}
\usepackage{float}
\usepackage{wrapfig}
\usepackage[inline]{enumitem}
\usepackage{acronym}

\newtheorem{theorem}{Theorem}
\newtheorem{lemma}[theorem]{Lemma}
\newtheorem{proposition}[theorem]{Proposition}
\newtheorem{corollary}[theorem]{Corollary}
\newtheorem{definition}[theorem]{Definition}
\newtheorem{remark}[theorem]{Remark}

\newcommand{\RR}{\mathbb{R}}
\newcommand{\EE}{\mathbb{E}}

\newcommand{\KL}{D_{\mathrm{KL}}}
\newcommand{\TV}{\mathrm{TV}}

\DeclareMathOperator*{\argmin}{arg\,min}

\DeclareMathOperator{\topB}{Top\text{-}B}

\newcommand{\tinymetric}[1]{{\scriptsize #1}}

\usepackage{needspace}
\makeatletter
\newsavebox{\SBS@tablebox}
\newcommand{\SideTableParagraph}[4]{%
  \par
  \sbox{\SBS@tablebox}{%
    \begin{minipage}[t]{#2}%
      \vspace{0pt}%
      #3%
    \end{minipage}%
  }%
  \Needspace{\dimexpr\ht\SBS@tablebox+\dp\SBS@tablebox+2\baselineskip\relax}%
  \begingroup
    \dimen@=\dimexpr\linewidth-\wd\SBS@tablebox-1em\relax
    \def\SBS@shape{}%
    \@tempcnta=\z@
    \loop\ifnum\@tempcnta<#1
      \xdef\SBS@shape{\SBS@shape 0pt \the\dimen@}%
      \advance\@tempcnta\@ne
    \repeat
    \xdef\SBS@shape{\SBS@shape 0pt \the\linewidth}%
    \parshape=\numexpr#1+1\relax \SBS@shape
    \noindent
    \makebox[0pt][l]{%
      \hspace*{\dimexpr\linewidth-\wd\SBS@tablebox\relax}%
      \raisebox{\dimexpr\ht\strutbox-\ht\SBS@tablebox\relax}[0pt][0pt]{%
        \usebox{\SBS@tablebox}%
      }%
    }%
    \tolerance=1500\emergencystretch=1em
    #4\par
  \endgroup
}
\makeatother

\makeatletter
\renewcommand{\paragraph}{\@startsection{paragraph}{4}{\z@}%
  {0.25ex plus 0.1ex minus 0.05ex}%
  {-1em}%
  {\normalfont\normalsize\bfseries}}
\makeatother

\acrodef{GR}{generative retrieval}
\acrodef{MGR}{multimodal generative retrieval}
\acrodef{RQ}{residual quantization}

\title{Closing the Indexing-Decoding Gap in Multimodal Generative Retrieval via Prefix Retention Optimization}

\author{%
\textbf{Yufei Chen$^{1}$ \quad
Zihan Wang$^{2}$ \quad
Yubao Tang$^{2}$} \\
\textbf{Yukun Zhao$^{1}$ \quad
Maarten de Rijke$^{2}$ \quad
Zhaochun Ren$^{3}$} \\
$^{1}$Shandong University \quad
$^{2}$University of Amsterdam \\
$^{3}$Leiden University \\
\texttt{cyf200409@gmail.com, zhw.cypher@gmail.com} \\
\texttt{y.li7@uva.nl, zhaoyukun@sdu.edu.cn, m.derijke@uva.nl} \\
\texttt{z.ren@liacs.leidenuniv.nl}
}

\begin{document}

\maketitle

\begin{abstract}
\Acl{MGR} formulates multimodal retrieval as discrete identifier generation, eliminating the need for explicit similarity search over external embeddings.
Existing approaches construct identifiers via residual quantization and decode them with trie-constrained beam search.
This combination introduces an \emph{indexing-decoding gap}: 
identifier learning objectives, including reconstruction and contrastive losses, do not explicitly enforce prefix discriminability during decoding.
As a result, even well-optimized identifiers can be irreversibly pruned early in beam search due to low-rank prefixes.
We theoretically characterize this gap and derive a survival bound that relates prefix retention to three controllable factors in indexing and decoding.
Building on this bound, we propose \textbf{PRO}, \emph{prefix retention optimization}, a unified framework comprising three  mechanisms: 
\begin{enumerate*}[label=(\roman*)]
\item prefix ranking distillation aligns quantized prefix rankings with those induced by pre-quantization embeddings using a listwise loss;
\item vocabulary scheduling increases codebook sizes from shallow to deep residual quantization levels to reduce early competition from non-target prefixes; and 
\item geometric score fusion vectorizes each candidate prefix and incorporates its similarity to the query into beam search scoring, further reducing the indexing–decoding mismatch.
\end{enumerate*}
Experiments on nine multimodal retrieval tasks show that PRO improves retention of target identifier prefixes and outperforms existing multimodal generative retrieval baselines.\footnote{Our code is available at \url{https://github.com/layingfish/MGR_PRO}.}
\end{abstract}

\section{Introduction}
\label{sec:intro}
Multimodal information retrieval requires matching queries and targets across modalities such as text, images, and composed inputs.
Dense retrieval dominates this setting by comparing learned continuous embeddings, but it requires maintaining and searching external vector indexes~\citep{wei2023uniir}.
\Ac{GR} offers an alternative: a sequence-to-sequence model directly generates a discrete identifier for the target item~\citep{decao2021autoregressive,tay2022transformer}.
This paradigm was first studied in text retrieval and recommender systems~\citep{tay2022transformer, rajput2023recommender}, and has recently been extended to multimodal search~\citep{zhang2024irgen,li2024grace,genius2025}.

Existing \ac{MGR} methods, including IRGen~\citep{zhang2024irgen}, GRACE~\citep{li2024grace}, AVG~\citep{cai2025avg}, and SemCORE~\citep{wu2025semcore}, focus on individual cross-modal retrieval tasks.
GENIUS~\citep{genius2025} takes a step further by introducing a unified generative framework for diverse multimodal retrieval tasks, achieving state-of-the-art performance in this setting.
It consists of three stages:
(i) an indexing stage that uses a unified multimodal encoder to map queries and target items into a shared continuous embedding space, and quantizes item embeddings into multi-level \ac{RQ} codes as identifiers;
(ii) a retrieval training stage that trains an autoregressive decoder to generate these identifiers from queries; and
(iii) a decoding stage where trie-constrained beam search ensures that the decoder only produces valid candidate identifiers during inference~\citep{decao2021autoregressive}.
Despite its versatility, GENIUS exhibits a structural \emph{indexing--decoding gap}, as illustrated in Fig.~\ref{fig:idg}: during indexing, identifier learning objectives, including reconstruction and contrastive losses, are optimized independently of decoding and therefore do not explicitly enforce prefix discriminability for decoding.
As a result, prefix survival is difficult to guarantee because trie-constrained beam search expands prefixes sequentially and irrevocably prunes low-ranked candidates at each level.


\begin{wrapfigure}[31]{r}{0.50\textwidth}
\centering
\vspace{-22pt}
\includegraphics[width=0.47\textwidth]{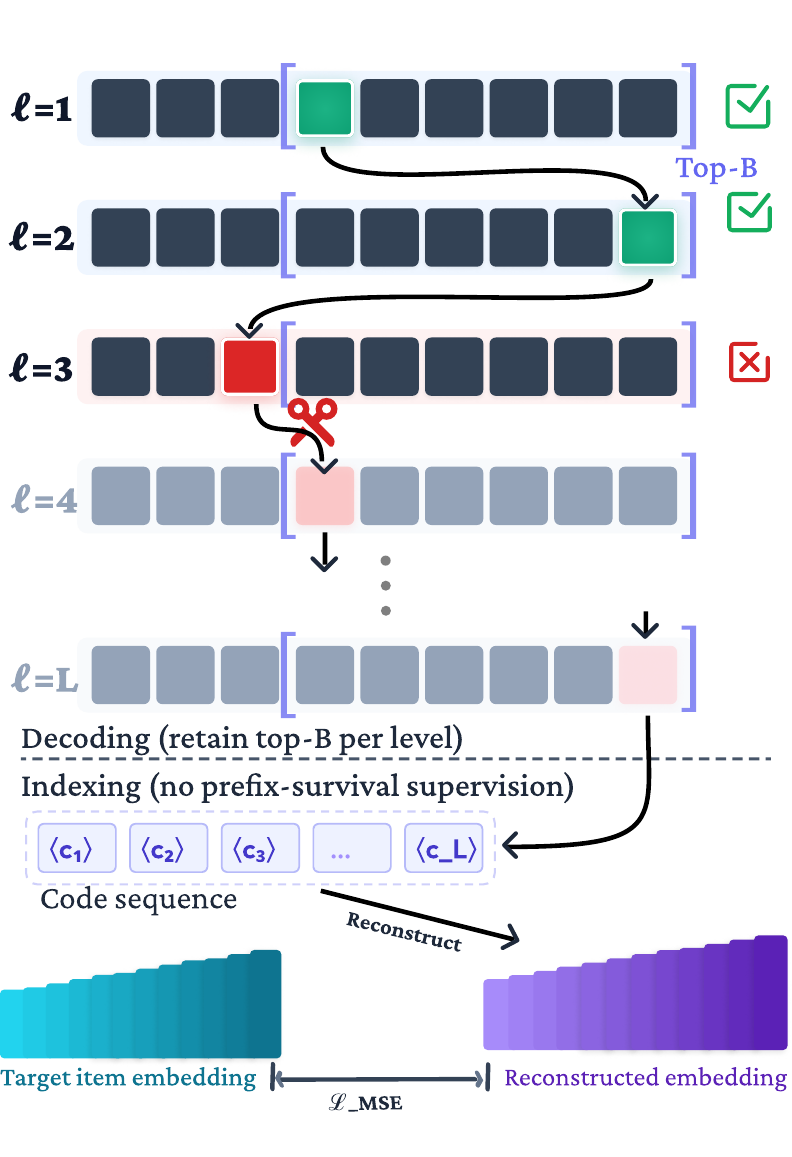}
\vspace{-6pt}
\caption{The indexing-decoding gap.
\textbf{Top} (decoding): Beam search prunes irrevocably at each level; the target item (colored path) survives levels~1--2 but is lost at level~3.
\textbf{Bottom} (indexing): Existing identifier learning objectives do not explicitly supervise prefix survival during decoding.}
\label{fig:idg}
\vspace{-10pt}
\end{wrapfigure}

In this paper, we theoretically characterize this gap and show that existing indexing objectives are insufficient to determine whether a target identifier prefix survives beam search: two tokenizers can achieve the same objective values while one preserves the target prefix and the other prunes it under decoding (see Thm.~\ref{thm:qdg}).
We further derive a survival bound that relates prefix survival to three controllable factors across indexing and decoding (see Thm.~\ref{thm:survival}).
Building on this bound, we propose \textbf{PRO}, the \emph{prefix retention optimization} framework, to improve prefix survival across indexing and decoding by targeting the controllable factors.
During indexing, prefix ranking distillation enforces consistency between prefix rankings of RQ partial reconstructions and those induced by pre-quantization embeddings via a listwise objective, while vocabulary scheduling progressively increases codebook sizes across RQ levels to mitigate early competition from non-target prefixes.
During decoding, geometric score fusion augments beam search scoring with query–prefix similarity derived from reconstructed prefixes, reducing the mismatch between indexing and decoding.

We evaluate PRO on nine multimodal retrieval tasks spanning text-to-image, image-to-text, and composed retrieval. Across all settings, PRO consistently improves prefix survival and outperforms state-of-the-art MGR baselines. Extensive results further demonstrate strong generalizability, with robust gains across diverse multimodal encoders and consistent improvements on a representative RQ-based recommendation model, TIGER~\citep{rajput2023recommender}.


Our contributions are:
\begin{enumerate*}[label=(\roman*)]
  \item We identify and formalize the indexing–decoding gap in \ac{MGR} as a mismatch between identifier learning objectives and sequential beam search pruning.
  \item We show that existing identifier learning objectives do not determine prefix survival, and derive a bound relating it to three controllable factors across indexing and decoding.
  \item We propose PRO, a framework for improving prefix survival via prefix ranking distillation, vocabulary scheduling, and geometric score fusion.
  \item Experiments on nine multimodal retrieval tasks show that PRO improves prefix survival, consistently surpasses state-of-the-art MGR methods, and demonstrates strong generalizability across multimodal encoders and RQ-based methods.
\end{enumerate*}

\section{Preliminaries}
\label{sec:preliminaries}

To facilitate the theoretical analysis of the indexing--decoding gap and the proposed PRO, we briefly introduce the \ac{MGR} framework proposed by GENIUS~\citep{genius2025}. This framework consists of three stages: indexing, retrieval training, and decoding.

\textbf{Indexing stage.}
\label{sec:prelim}
The indexing stage constructs discrete identifiers for items by progressively quantizing their representations, which are extracted by a frozen multimodal encoder. These identifiers serve as supervision for subsequent decoder training.

\emph{Multimodal encoder.}
The encoder maps queries and targets from heterogeneous modalities, including images, texts, and image-text compositions, into a shared $d$-dimensional embedding space. It is trained with a contrastive loss and kept frozen during both identifier generation and decoder training.

\emph{Residual quantization.}
Let $q \in \RR^d$ and $x_j \in \RR^d$ denote the pre-quantization fused query and target embeddings, respectively. We use $z_j^\ell$ to denote the RQ code index of item~$j$ at level~$\ell$.
The tokenizer maps each target embedding $x_j$ to a discrete identifier via an $L$-level residual quantizer. Specifically, it (i) initializes the residual as $r_0 = x_j$; (ii) at each level $\ell$, selects the nearest codeword from the level-specific codebook $\{c_{\ell,1}, \ldots, c_{\ell,V_\ell}\} \subset \RR^d$ by
$z_j^\ell = \argmin_{v \in \{1,\ldots,V_\ell\}} \|r_{\ell-1} - c_{\ell,v}\|^2$; and (iii) updates the residual as $r_\ell = r_{\ell-1} - c_{\ell,z_j^\ell}$. Here, $V_\ell$ denotes the vocabulary size at level~$\ell$. The resulting code sequence $(z_j^1,\ldots,z_j^L)$ forms the discrete identifier of item~$j$, with partial reconstruction $\hat{x}_{j,\ell}=\sum_{t=1}^{\ell} c_{t,z_j^t}$ and full reconstruction $\hat{x}_j=\hat{x}_{j,L}$; applying the same residual quantizer to $q$ gives the full quantized query vector $\hat{q}=\sum_{\ell=1}^{L} c_{\ell,z_q^\ell}$, where $z_q^\ell$ is selected by the same residual-quantization rule initialized at~$q$.
Following GENIUS~\citep{genius2025}, we adopt a modality-decoupled first code level to distinguish target modalities: the first-level codebook contains three codes (one per modality), while subsequent levels are shared to capture semantic content in the unified embedding space.

\emph{Training objective.}
The standard RQ tokenizer is trained with three objectives: (i) an RQ loss that stabilizes codebook learning by aligning residuals with selected codewords (updated via exponential moving average); (ii) a mean squared error (MSE) loss aligning quantized query and target vectors, $\mathcal{L}_{\mathrm{mse}}=\|\hat{q}-\hat{x}_j\|_2^2$; and (iii) a contrastive loss~\citep{khosla2020supervised} on pre-quantization embeddings, with one direction given by
$\mathcal{L}_{\mathrm{cl}}=-\log \frac{\exp(\langle q,x_j\rangle/\tau)}{\sum_k \exp(\langle q,x_k\rangle/\tau)}$.
During training, the same RQ codebooks are applied to both $q$ and $x_j$ to obtain $\hat{q}$ and $\hat{x}_j$, while only the target code sequence is stored as the item identifier. Here, $x_k$ denotes a positive target or a contrastive negative, and $\tau>0$ is the temperature. The standard RQ loss is given in App.~\ref{app:impl}.
Overall, these objectives constrain local residual-to-codeword fitting, full quantized vectors, and pre-quantization embeddings, but do not explicitly regulate query-conditioned rankings of intermediate prefixes or partial reconstructions $\hat{x}_{j,\ell}$.

\textbf{Retrieval training and decoding.}
After tokenizer training, the discrete identifiers supervise an autoregressive decoder.
(i) \emph{Training.} The decoder~\citep{raffel2020exploring} maps query embeddings to code sequences using a cross-entropy loss over the $L$ tokens.
(ii) \emph{Decoding.} Candidate identifiers are generated via trie-constrained beam search~\citep{decao2021autoregressive} with beam width~$B$, where decoding is restricted to valid identifiers and beams are ranked by cumulative log-probability.

\section{The Indexing-Decoding Gap}
\label{sec:qdg}

In this section, we provide a theoretical analysis of the \emph{indexing-decoding gap} between identifier learning and target item decoding.
We first establish a negative result: existing tokenizer objectives used for identifier learning, including contrastive, residual fitting, and reconstruction losses, do not determine whether the target prefix survives beam search.
For theoretical analysis, we simulate beam search scoring with a rule that ranks each prefix by the similarity between the query and its RQ partial reconstruction, $\langle q, \hat{x}_{j,\ell} \rangle$.
We refer to this scoring rule as the \emph{quantized oracle}.
This oracle evaluates prefixes using partial reconstructions induced by RQ codes in the quantized representation space.

\begin{theorem}

\label{thm:qdg}

For any beam width $B \geq 1$ and RQ levels $L \geq 2$, there exist a query~$q$, a set of $B{+}1$ items, and two $L$-level residual quantizers $\mathcal{Q}^+$ and $\mathcal{Q}^-$ 
such that:
\begin{enumerate*}[label=\textup{(\roman*)}]

  \item $\mathcal{Q}^+$ and $\mathcal{Q}^-$ share the same pre-quantization query and item embeddings, produce identical greedy code assignments, and achieve identical values for all training objectives, including contrastive, residual-fitting, and reconstruction losses computed from full reconstructions.
  \item Under trie-constrained beam search over valid item identifiers with width~$B$ using quantized oracle scores, the target prefix survives at every level in $\mathcal{Q}^+$, whereas it is pruned at level~1 in $\mathcal{Q}^-$.

\end{enumerate*}

\end{theorem}


\emph{Proof sketch.} We construct two RQ instances $\mathcal{Q}^+$ and $\mathcal{Q}^-$ over the same query and $B{+}1$ items.
First, both quantizers share the same pre-quantization embeddings, select the same RQ code index at every level under the greedy nearest codeword rule, and have identical values for the existing tokenizer objectives.
The only difference is in the first two codewords selected by the target item.
In $\mathcal{Q}^+$, the first selected codeword has a positive component along the query direction, so the target prefix receives the highest level 1 quantized oracle score.
In $\mathcal{Q}^-$, we reverse this query direction component in the first selected codeword, which lowers the target's level 1 prefix score, and add the opposite compensation to the second selected codeword, so that the sum of the first two selected codewords and therefore the full reconstruction remain unchanged.
Consequently, the existing tokenizer objectives considered in the theorem assign the same values to the two constructions.
As a result, the target prefix survives under $\mathcal{Q}^+$ but is ranked below all $B$ competitors and pruned at level~1 under $\mathcal{Q}^-$.
Full details are provided in App.~\ref{app:proofs}. If the first code level is reserved for a modality token, the construction applies from levels~2 and~3.

Thm.~\ref{thm:qdg} shows that existing tokenizer objectives cannot guarantee prefix discriminability during beam search.
Although later RQ levels can compensate early errors in the final reconstruction, beam search is irreversible and cannot recover a target once its prefix has been pruned.
Therefore, prefix discriminability at intermediate levels must be explicitly considered beyond optimizing complete identifiers.
This is analogous to the failure of successive refinement in source coding~\citep{equitz1991successive}: a code optimized for final distortion need not be useful at intermediate stages.

\emph{Empirical analysis.}
\phantomsection
\label{sec:evidence}
We examine whether the failure mode revealed by Thm.~\ref{thm:qdg} also appears in practice.
We measure the per-level survival rate of the target identifier token under quantized oracle scoring on COCO~\citep{lin2014microsoft,wei2023uniir} with a \ac{RQ}-based tokenizer~\citep{lee2022autoregressive} ($L{=}16$, $V{=}4096$, $B{=}50$, 5{,}000 queries).
At each \ac{RQ} level, this metric conditions on the preceding target prefix and checks whether the target token at that level is ranked among the top-$B$ candidate tokens.
As shown in Fig.~\ref{fig:mse_survival}\subref{fig:prefix_survival_decay}, the target token becomes much harder to retain after the first \ac{RQ} level.
This shows the same failure mode as Thm.~\ref{thm:qdg}: due to insufficient prefix discriminability, target identifier tokens are easily lost during beam decoding.
We further use complete reconstruction MSE as a diagnostic of reconstruction quality.
Fig.~\ref{fig:mse_survival}\subref{fig:mse_survival_ecdf} shows that the ECDF curves of survived and pruned items nearly overlap.
Fig.~\ref{fig:mse_survival}\subref{fig:mse_survival_quintile} groups items by MSE quintiles and shows that the prefix survival rate varies little across different MSE ranges.
These results suggest that good complete reconstruction alone does not predict whether a target survives beam search.

\begin{figure}[t]
  \centering
  \captionsetup[subfigure]{skip=1pt}
  \begin{subfigure}[t]{0.315\linewidth}
    \centering
    \includegraphics[width=\linewidth]{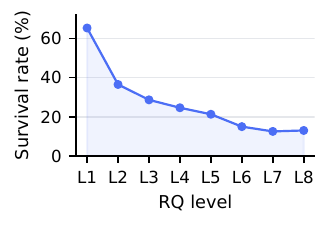}
    \caption{Per-level token survival}
    \label{fig:prefix_survival_decay}
  \end{subfigure}\hfill
  \begin{subfigure}[t]{0.315\linewidth}
    \centering
    \includegraphics[width=\linewidth]{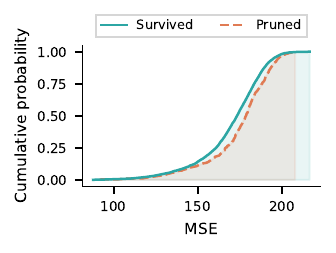}
    \caption{MSE: survived vs. pruned}
    \label{fig:mse_survival_ecdf}
  \end{subfigure}\hfill
  \begin{subfigure}[t]{0.315\linewidth}
    \centering
    \includegraphics[width=\linewidth]{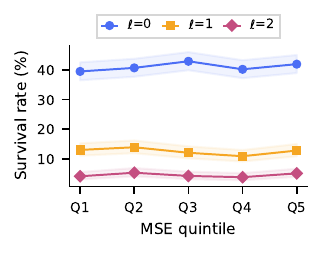}
    \caption{Survival by MSE quintile}
    \label{fig:mse_survival_quintile}
  \end{subfigure}
    \caption{
            Empirical diagnostics of the indexing--decoding gap on COCO under quantized oracle scoring.
            (a) The per-level survival rate of the target identifier token drops sharply across early RQ code positions.
            (b,c) Full reconstruction MSE weakly predicts prefix survival: survived/pruned
            empirical cumulative distribution functions (ECDFs) nearly overlap, and survival remains flat across MSE quintiles.
            }
  \label{fig:mse_survival}
\end{figure}


\section{Prefix Retention Optimization}
\label{sec:method}

\textit{Under what conditions does the target prefix survive} and \textit{how can these conditions be formulated as continuous quantities amenable to optimization}?
To answer those questions, we first derive a survival bound that characterizes prefix survival through three quantities: ranking divergence, teacher margin, and decoder mismatch.
The bound provides a constructive counterpart to Thm.~\ref{thm:qdg}: while existing objectives cannot determine survival, survival is guaranteed when the errors introduced by indexing and decoding are sufficiently small relative to the teacher margin.
Guided by this bound, we propose \emph{Prefix Retention Optimization (PRO)}, which reduces ranking divergence, enlarges the teacher margin, and mitigates decoder mismatch via three corresponding components.

\textbf{Survival bound.}
\phantomsection
\label{sec:survival}
We begin with prefix-level notation. 
For a target item~$j$ with RQ code sequence $z_j^{1:L}$, let $\hat{x}_{j,\ell} := \sum_{t=1}^{\ell} c_{t,z_j^t}$ denote its partial reconstruction after $\ell$ levels, so that
$\hat{x}_{j,L} = \hat{x}_j$. Let $\mathcal{P}_\ell$ denote the set of
distinct length-$\ell$ code prefixes, and let $p \in \mathcal{P}_\ell$ denote a generic prefix.
Let $p^+_\ell := z^{1:\ell}$ denote the target prefix at level~$\ell$.
We analyze each beam step through distributions over $\mathcal{P}_\ell$: 
a teacher distribution~$\pi^*_\ell$ defined by pre-quantization embedding similarities, 
a \emph{quantized oracle distribution}~$\bar{\pi}_\ell$ induced by partial reconstructions, 
and a \emph{decoder prefix distribution}~$\tilde{\pi}_\ell$ derived from decoder scores. 
All distributions are conditioned on the query embedding~$q$, and their formal definitions are given in App.~\ref{app:proofs}.

The survival of the ground-truth prefix $p^+_\ell$ is governed by three quantities:
\begin{align}
  K_\ell
    &= \KL\bigl(\pi^*_\ell \,\|\, \bar{\pi}_\ell\bigr),
    &\quad&\text{(ranking divergence)}
    \label{eq:ranking_kl} \\
  m_\ell
    &= \pi^*_\ell(p^+_\ell)
       - \pi^*_{\ell,B+1},
    &&\text{(teacher margin)}
    \label{eq:margin} \\
  \varepsilon_\ell
    &= \TV\bigl(\bar{\pi}_\ell ,\,
       \tilde{\pi}_\ell\bigr),
    &&\text{(decoder mismatch)}
    \label{eq:dec_mismatch}
\end{align}
where $\pi^*_{\ell,B+1}$ denotes the $(B{+}1)$-th largest probability under the teacher distribution (with $V_\ell > B$ in practice).
Thus, $K_\ell$ measures the divergence between teacher and quantized rankings, $m_\ell$ is the top-$B$ margin under the teacher, and $\varepsilon_\ell$ captures the discrepancy between oracle and decoder distributions.

\begin{theorem}
\label{thm:survival}
Fix a query $q$ and beam width $B$.
If, for every level $\ell \in \{1,\dots,L\}$,
\begin{equation}
  \underbrace{\sqrt{\tfrac{1}{2}\,K_\ell(q)}
    \vphantom{\Big|}}_{\text{\normalsize ranking divergence}}
  \;+\;
  \underbrace{\varepsilon_\ell(q)
    \vphantom{\Big|}}_{\text{\normalsize decoder mismatch}}
  \;\;<\;\;
  \underbrace{\tfrac{1}{2}\,m_\ell(q)
    \vphantom{\Big|}}_{\text{\normalsize teacher margin}},
\label{eq:survival}
\end{equation}
then the target prefix remains in the beam through all $L$ levels.
\end{theorem}

When $m_\ell \leq 0$, the bound provides no guarantee, as the target prefix is not ranked within the teacher top-$B$.
All distributions are defined over the full prefix set~$\mathcal{P}_\ell$; thus, being top-$B$ globally implies being top-$B$ within any subset containing~$p^+_\ell$ (see App.~\ref{app:proofs}).
The proof proceeds in two steps.
First, if quantization perturbs the teacher distribution by less than half of the teacher margin, the target prefix remains ahead of all prefixes outside the teacher top-$B$ under the quantized oracle.
Pinsker's inequality~\citep{cover2006elements} converts this perturbation into the term $\sqrt{K_\ell/2}$.
Second, applying the triangle inequality to incorporate the shift from the oracle to the decoder distribution yields the additional term $\varepsilon_\ell$.
Hence, survival is guaranteed when the combined deviation from quantization and decoding is sufficiently small relative to the teacher margin.
The value of Thm.~\ref{thm:survival} lies in the design principle it exposes: reduce $K_\ell$, reduce $\varepsilon_\ell$, and enlarge $m_\ell$ at levels where beam pruning is irreversible.
Although the bound is conservative, as is typical of Pinsker-based analyses~\citep{fedotov2003refinements,rioul2024historical,sason2016f}, it correctly predicts the monotonic relationship between $K_\ell$ and prefix survival observed in our experiments (see Sec.~\ref{sec:main_results}).
This principle directly motivates the three components of PRO.

\begin{figure}[t]
\centering
\includegraphics[width=0.97\textwidth]{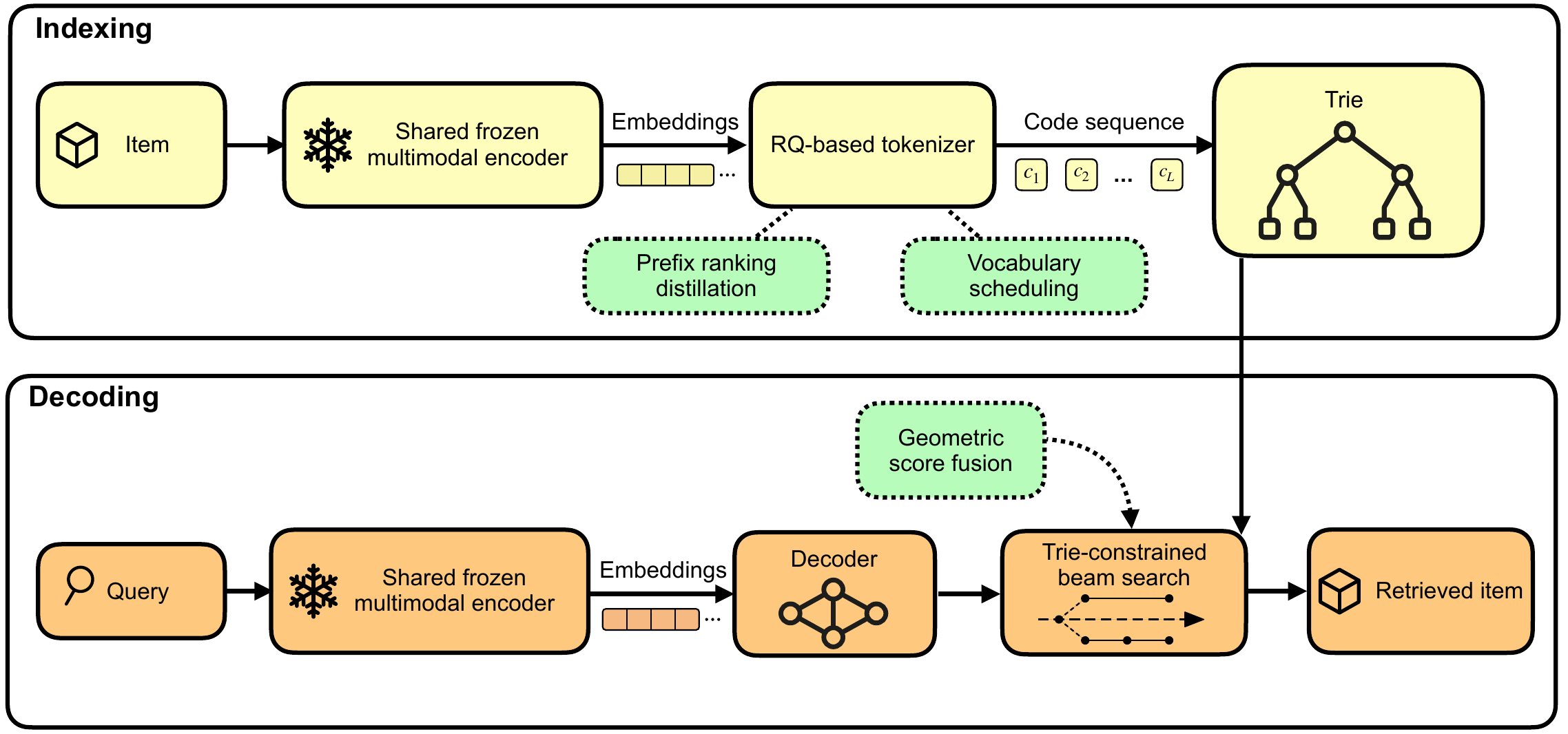}
\caption{Overview of PRO. During tokenizer training, \textbf{prefix ranking distillation} reduces the ranking divergence $K_\ell$ and \textbf{vocabulary scheduling} enlarges the teacher margin $m_\ell$.
During inference, \textbf{geometric score fusion} helps reduce the decoder mismatch $\varepsilon_\ell$.}
\label{fig:darq}
\vspace*{-.5\baselineskip}
\end{figure}

\textbf{The PRO Framework.}
\phantomsection
\label{sec:framework}
As illustrated in Fig.~\ref{fig:darq}, PRO integrates two tokenizer-side components with one inference-time component.
Prefix ranking distillation and vocabulary scheduling operate on identifier learning, while geometric score fusion augments scoring in trie-constrained beam search.

\emph{Prefix ranking distillation.}
\phantomsection
\label{sec:rpq}
Prefix ranking distillation is designed to reduce the ranking divergence $K_\ell$.
For each query and each RQ level, we compare two distributions over the same candidate prefixes.
The teacher distribution scores prefixes using the pre-quantization item embeddings, while the quantized distribution scores the same prefixes using the partial reconstructions produced by the RQ codebooks.
If these two distributions disagree, the target prefix can be ranked differently after indexing.
We therefore add a listwise KL distillation loss~\citep{hinton2015distilling,cao2007learning} during RQ training:
\begin{equation}
  \mathcal{L}_{\mathrm{rank}}
  \;=\; \EE_{q}\!\left[\,
    \frac{1}{L}
    \sum_{\ell=1}^{L}
    \KL\bigl(\pi^*_\ell \,\|\, \bar{\pi}_\ell\bigr)
  \right].
\label{eq:rpq}
\end{equation}
Both distributions are approximated over an in-batch candidate set with shared temperature~$\tau$.\footnote{In-batch approximation is a common practice in distillation-based training and provides a reasonable estimate of $K_\ell$ when the batch size is sufficiently large.}
Minimizing this loss directly reduces the empirical estimate of $K_\ell$.
It encourages the partial reconstructions at each RQ level to preserve the prefix probability distribution induced by the pre-quantization embeddings for the same query.
The ranking distillation loss is incorporated as an auxiliary term in the tokenizer training objective:
\begin{equation}
  \mathcal{L}_{\mathrm{total}}
  \;=\; \mathcal{L}_{\mathrm{cl}}
  \;+\; \beta_{\mathrm{rq}}\,\mathcal{L}_{\mathrm{rq}}
  \;+\; \gamma_{\mathrm{mse}}\,\mathcal{L}_{\mathrm{mse}}
  \;+\; \lambda\,\mathcal{L}_{\mathrm{rank}}\,,
\label{eq:stage1}
\end{equation}
where $\mathcal{L}_{\mathrm{rq}}$ denotes the RQ loss, and $\beta_{\mathrm{rq}}$, $\gamma_{\mathrm{mse}}$, and $\lambda$ control the strengths of the RQ, MSE, and ranking distillation losses.

\emph{Vocabulary scheduling.}
\phantomsection
\label{sec:daca}
Vocabulary scheduling is designed to enlarge the teacher margin $m_\ell$.
At shallow RQ levels, beam search makes pruning decisions using only short prefixes.
If the early codebook is large, the target prefix must compete with many non-target prefixes before deeper codewords can refine the identifier.
A common RQ-VAE design~\citep{lee2022autoregressive} uses the same codebook size across residual levels, and product quantization~\citep{jegou2011product} quantizes separate subspaces with separate subcodebooks, often using the same number of codewords per subspace.
We therefore use an ascending schedule
$V_1 \leq V_2 \leq \cdots \leq V_L$, allocating less capacity to early levels where pruning is irreversible.
A smaller shallow vocabulary reduces the number of early competing prefixes.
Under the teacher distribution, probability mass is spread over fewer early prefixes, increasing the probability gap between the target prefix and the strongest prefix just outside the beam.
Thus, vocabulary scheduling enlarges the teacher margin at the levels where pruning is most irreversible.
The larger deeper vocabularies preserve representation capacity after the early survival bottleneck has been alleviated.
Tab.~\ref{tab:daca} compares three scheduling strategies:
descending, uniform, and ascending.
The ascending schedule consistently yields the largest
$\bar{m}_\ell$ and the highest survival rates, while the descending schedule performs worst.
Hence, reducing early-level competition is an effective way to enlarge the teacher margin in practice, despite the smaller early-level vocabulary reducing reconstruction capacity:
reconstruction fidelity and beam search survival
are distinct objectives.

\begin{table}[t]
\caption{Performance of three vocabulary schedules on COCO.}
\label{tab:daca}
\centering
\setlength{\tabcolsep}{2mm}
\resizebox{0.95\textwidth}{!}{%
\begin{tabular}{@{}l ccc ccc ccc ccc@{}}
\toprule
& \multicolumn{3}{c}{Level 0}
& \multicolumn{3}{c}{Level 1}
& \multicolumn{3}{c}{Level 2}
& \multicolumn{3}{c}{Level 3} \\
\cmidrule(lr){2-4} \cmidrule(lr){5-7} \cmidrule(lr){8-10} \cmidrule(lr){11-13}
& Desc. & Unif. & \textbf{Asc.}
& Desc. & Unif. & \textbf{Asc.}
& Desc. & Unif. & \textbf{Asc.}
& Desc. & Unif. & \textbf{Asc.} \\
\midrule
$V$
  & 2048 & 1024 & \textbf{512}
  & 2048 & 1024 & \textbf{512}
  & 2048 & 1024 & \textbf{512}
  & 2048 & 1024 & \textbf{512} \\
$\bar{m}_\ell$ $\uparrow$
  & 0.035 & 0.079 & \textbf{0.109}
  & 0.005 & 0.013 & \textbf{0.022}
  & 0.003 & 0.007 & \textbf{0.013}
  & 0.002 & 0.005 & \textbf{0.010} \\
$\bar{K}_\ell$ $\downarrow$
  & 1.04 & 1.02 & \textbf{0.92}
  & 0.62 & 0.58 & \textbf{0.51}
  & 0.46 & 0.43 & \textbf{0.35}
  & 0.39 & 0.35 & \textbf{0.28} \\
Surv.\,(\%) $\uparrow$
  & 55.2 & 66.6 & \textbf{75.1}
  & 19.5 & 31.2 & \textbf{41.1}
  & 16.4 & 23.9 & \textbf{31.1}
  & 10.0 & 18.0 & \textbf{26.2} \\
\bottomrule
\end{tabular}%
}
\vspace*{-.5\baselineskip}
\end{table}

\emph{Geometric score fusion.}
\phantomsection
\label{sec:cdsb}
Geometric score fusion is designed to reduce the decoder mismatch $\varepsilon_\ell$.
During beam search, the decoder assigns a log-probability to each legal next token based on the query and the current prefix.
Each legal token also corresponds to an RQ codeword.
Let $\hat{x}_p$ denote the partial reconstruction represented by the current beam prefix~$p$, and let $r_p = q-\hat{x}_p$ denote the residual from this partial reconstruction to the query embedding.
This residual is defined only for scoring candidate prefixes during decoding and should not be confused with the RQ residual used during tokenizer assignment.
Since the decoder score does not explicitly measure whether a candidate codeword moves the partial reconstruction closer to the query, we use $r_p$ to compute a geometric bias for each candidate token:
\begin{equation}
  s^{\mathrm{fuse}}_\ell(v \mid p, q)
  \;=\; \log P_{\mathrm{decoder}}(v \mid q, p)
  \;+\; \omega \bigl(
    2 r_p^\top c_{\ell,v}
    - c_{\ell,v}^\top c_{\ell,v}
  \bigr),
\label{eq:cdsb}
\end{equation}
Here, $c_{\ell,v}$ is the codeword associated with candidate token~$v$, and $\omega$ controls the strength of geometric score fusion.
The added term is exactly the reduction in squared distance from the query to the partial reconstruction after appending~$v$, since $\|r_p\|^2 - \|r_p-c_{\ell,v}\|^2 = 2 r_p^\top c_{\ell,v} - c_{\ell,v}^\top c_{\ell,v}$.
Thus, the bias rewards candidate codewords that reduce the squared query-to-reconstruction distance, with larger values corresponding to larger reductions.
After adding this bias to the decoder log-probability, beam search can take the RQ reconstruction geometry into account when scoring candidate prefixes.
This better aligns the fused decoder distribution with the quantized distribution and helps reduce decoder mismatch.
App.~\ref{app:score_fusion} formalizes this connection by defining the corresponding conditional distributions and showing when the fused distribution moves closer to the quantized oracle distribution.
At $\omega = 0$, the score reduces to the original decoder score.
Score fusion operates at inference only and requires no changes to tokenizer training.

\section{Experimental Settings}
\label{sec:exp}


\label{sec:setup}

We evaluate on 9 dataset--task pairs from the M-BEIR benchmark~\citep{wei2023uniir} and Flickr30k~\citep{young2014flickr30k}, covering three query types: single-modality (COCO~\citep{lin2014microsoft}, Flickr30k~\citep{young2014flickr30k}, WebQA~\citep{chang2022webqa}, NIGHTS~\citep{fu2023dreamsim}), text-to-multimodal, and composed image+text (CIRR~\citep{liu2021cirr}, OVEN~\citep{hu2023oven}, InfoSeek~\citep{chen2023infoseek}). Candidate pools range from 1K to 612K. All results use task-specific pools and report Recall@1 and Recall@5 without reranking.
We compare with two dense baselines from UniIR~\citep{wei2023uniir}, CLIP-SF and BLIP-FF, and five generative methods: IRGen~\citep{zhang2024irgen}, GRACE~\citep{li2024grace}, AVG~\citep{cai2025avg}, SemCORE~\citep{wu2025semcore}, and GENIUS~\citep{genius2025}. Among them, GENIUS is the only universal multimodal generative retriever. The remaining four methods support only text-to-image retrieval.
For fairness, GENIUS and PRO share the same CLIP-SF encoder and T5-small decoder~\citep{raffel2020exploring}. For prior single-task baselines, we report published or reproduced results when available. PRO uses 16 RQ levels, an ascending vocabulary schedule (512, 1024, 2048), prefix ranking distillation (weight 100, temperature 0.05), and geometric score fusion (weight 10) with beam size 50. Diagnostic analyses use beam size 20 to amplify pruning effects, while all reported metrics use beam size 50. Full training details are in App.~\ref{app:impl}.

\section{Experimental Results}
\textbf{Main experimental comparison.}
\phantomsection
\label{sec:main_results}
Tab.~\ref{tab:main} reports Recall@1 and Recall@5 across all $9$ tasks.
PRO consistently outperforms GENIUS on every task, indicating robust end-to-end improvements across diverse multimodal settings. 
Although dense baselines remain stronger overall, PRO significantly narrows the gap to dense retrieval. 
The reduction is most pronounced on Flickr30k text-to-image, where PRO closes over 60\% of the gap to dense retrieval, and on the two COCO tasks, where it closes nearly 50\%. 
These results demonstrate that PRO not only improves upon prior generative baselines but also substantially enhances the competitiveness of MGR relative to dense retrieval.

\begin{table}[!t]
\caption{Main results on M-BEIR, retrieved from a task-specific pool.
\textbf{Bold}: best in each category (dense / generative). A dash (--) indicates that the corresponding baseline does not support the setting.}
\label{tab:main}
\vspace{1mm}
\centering
\setlength{\tabcolsep}{3pt}
\resizebox{\textwidth}{!}{%
\footnotesize
\begin{tabular}{l cc cc cc cc cc cc cc cc cc}
\toprule
\multirow{3}{*}{\textbf{Method}}
& \multicolumn{2}{c}{COCO}
& \multicolumn{2}{c}{Flickr30k}
& \multicolumn{2}{c}{WebQA}
& \multicolumn{2}{c}{WebQA}
& \multicolumn{2}{c}{COCO}
& \multicolumn{2}{c}{NIGHTS}
& \multicolumn{2}{c}{CIRR}
& \multicolumn{2}{c}{OVEN}
& \multicolumn{2}{c}{InfoSeek} 
\\
& \multicolumn{2}{c}{$t {\to} i$}
& \multicolumn{2}{c}{$t {\to} i$}
& \multicolumn{2}{c}{$t {\to} t$}
& \multicolumn{2}{c}{$t {\to} i\text{,}t$}
& \multicolumn{2}{c}{$i {\to} t$}
& \multicolumn{2}{c}{$i {\to} i$}
& \multicolumn{2}{c}{$i\text{,}t {\to} i$}
& \multicolumn{2}{c}{$i\text{,}t {\to} i\text{,}t$}
& \multicolumn{2}{c}{$i\text{,}t {\to} t$} \\
\cmidrule(lr){2-3} \cmidrule(lr){4-5} \cmidrule(lr){6-7}
\cmidrule(lr){8-9} \cmidrule(lr){10-11} \cmidrule(lr){12-13}
\cmidrule(lr){14-15} \cmidrule(lr){16-17} \cmidrule(lr){18-19}
& R@1 & R@5
& R@1 & R@5
& R@1 & R@5
& R@1 & R@5
& R@1 & R@5
& R@1 & R@5
& R@1 & R@5
& R@1 & R@5
& R@1 & R@5 \\
\midrule
\multicolumn{19}{l}{\textit{Dense}} \\
CLIP-SF
  & \textbf{55.2} & \textbf{80.7}
  & \textbf{79.1} & \textbf{95.0}
  & \textbf{58.3} & \textbf{84.2}
  & 47.6 & 76.3
  & 65.7 & 87.5
  & \textbf{8.8} & \textbf{31.6}
  & \phantom{0}5.9 & 43.2
  & \textbf{49.4} & \textbf{69.2}
  & \textbf{14.5} & \textbf{28.8} \\
BLIP-FF
  & 53.7 & 79.6
  & 73.9 & 92.4
  & 51.9 & 78.6
  & \textbf{49.5} & \textbf{78.1}
  & \textbf{71.1} & \textbf{90.9}
  & 8.1 & 31.0
  & \textbf{24.4} & \textbf{50.9}
  & 36.1 & 56.4
  & 10.5 & 23.4 \\
\midrule
\multicolumn{19}{l}{\textit{Generative}} \\
IRGen
  & 29.6 & 50.7
  & 49.0 & 68.9
  & -- & --
  & -- & --
  & -- & --
  & -- & --
  & -- & --
  & -- & --
  & -- & -- \\
GRACE
  & 16.7 & 39.5
  & 37.4 & 59.5
  & -- & --
  & -- & --
  & -- & --
  & -- & --
  & -- & --
  & -- & --
  & -- & -- \\
AVG
  & 31.3 & 58.0
  & 62.8 & 85.4
  & -- & --
  & -- & --
  & -- & --
  & -- & --
  & -- & --
  & -- & --
  & -- & -- \\
SemCORE
  & 42.4 & 57.5
  & 69.0 & 83.0
  & -- & --
  & -- & --
  & -- & --
  & -- & --
  & -- & --
  & -- & --
  & -- & -- \\
GENIUS
  & 37.1 & 66.4
  & 56.8 & 81.3
  & 19.4 & 29.9
  & 29.1 & 51.3
  & 47.9 & 77.3
  & 1.2 & 11.6
  & \phantom{0}5.5 & 19.0
  & 26.4 & 36.0
  & \phantom{0}6.1 & 10.2 \\
Ours
  & \textbf{45.8} & \textbf{72.6}
  & \textbf{70.6} & \textbf{87.3}
  & \textbf{25.7} & \textbf{40.7}
  & \textbf{35.3} & \textbf{57.6}
  & \textbf{59.0} & \textbf{84.3}
  & \textbf{3.5} & \textbf{19.7}
  & \phantom{0}\textbf{8.6} & \textbf{27.8}
  & \textbf{34.5} & \textbf{51.3}
  & \phantom{0}\textbf{7.3} & \textbf{15.0} \\
\bottomrule
\end{tabular}%
}
\vspace*{-.5\baselineskip}
\end{table}

\textbf{Efficiency analysis.}
\phantomsection
\label{sec:efficiency}
We compare the retrieval efficiency of PRO and dense retrieval as the candidate pool size increases. 
As shown in Fig.~\ref{fig:efficiency}, PRO maintains nearly constant throughput as the pool grows, whereas dense retrieval degrades steadily, with the two curves crossing at around 50K candidates.
Moreover, geometric score fusion introduces negligible overhead, as the throughput curves with and without it almost overlap. 
These results indicate that PRO preserves the scalability advantage of generative retrieval while maintaining strong end-to-end performance. 
Further details, including index size breakdowns, are provided in App.~\ref{app:efficiency}.

\Needspace{22\baselineskip}
\noindent\textbf{Ablation study.}\phantomsection\label{sec:ablation}

\begin{wraptable}[17]{r}{0.38\columnwidth}
\vspace*{-2mm}
\footnotesize
\setlength{\tabcolsep}{2.5pt}
\caption{Ablation study result.}
\label{tab:ablation}
\vspace{-1mm}
\centering
\begin{tabular}{l l rrrr}
\toprule
& & Base & +PD & +VS & +SF \\
\midrule
\multirow{2}{*}{COCO \scalebox{.8}{$t{\to}i$}}
  & \tinymetric{R@1} & 36.8 & 40.2 & 42.1 & \textbf{46.2} \\
  & \tinymetric{R@5} & 66.5 & 68.2 & 70.7 & \textbf{73.0} \\
\addlinespace[1.5pt]
\multirow{2}{*}{COCO \scalebox{.8}{$i{\to}t$}}
  & \tinymetric{R@1} & 48.1 & 51.5 & 53.6 & \textbf{59.3} \\
  & \tinymetric{R@5} & 77.2 & 78.9 & 81.4 & \textbf{84.6} \\
\addlinespace[1.5pt]
\multirow{2}{*}{WebQA \scalebox{.8}{$t{\to}t$}}
  & \tinymetric{R@1} & 19.5 & 21.2 & 21.6 & \textbf{26.2} \\
  & \tinymetric{R@5} & 30.0 & 31.7 & 33.2 & \textbf{39.7} \\
\addlinespace[1.5pt]
\multirow{2}{*}{WebQA \scalebox{.8}{$t{\to}i\text{,}t$}}
  & \tinymetric{R@1} & 28.9 & 31.4 & 33.5 & \textbf{35.6} \\
  & \tinymetric{R@5} & 51.4 & 53.1 & 55.2 & \textbf{57.5} \\
\addlinespace[1.5pt]
\multirow{2}{*}{OVEN \scalebox{.8}{$i\text{,}t{\to}i$}}
  & \tinymetric{R@1} & 5.9  & 5.6  & 8.8  & \textbf{9.1}  \\
  & \tinymetric{R@5} & 14.7 & 14.2 & \textbf{19.0} & 17.6 \\
\addlinespace[1.5pt]
\multirow{2}{*}{OVEN \scalebox{.8}{$i\text{,}t{\to}i\text{,}t$}}
  & \tinymetric{R@1} & 26.7 & 28.4 & 32.2 & \textbf{34.4} \\
  & \tinymetric{R@5} & 36.4 & 39.9 & 47.5 & \textbf{51.4} \\
\bottomrule
\multicolumn{6}{l}{\scriptsize PD: prefix distill.\ ($K_\ell{\downarrow}$),
  VS: vocab sched.\ ($m_\ell{\uparrow}$),} \\
\multicolumn{6}{l}{\scriptsize SF: score fusion ($\varepsilon_\ell{\downarrow}$).
  Bold: best per row.}
\end{tabular}
\end{wraptable}

\noindent
To isolate the contribution of each component, we add them cumulatively following the survival bound (see Eq.~\ref{eq:survival}): prefix distillation (reducing ranking divergence, $K_\ell{\downarrow}$), vocabulary scheduling (increasing the teacher margin, $m_\ell{\uparrow}$), and score fusion (mitigating decoder mismatch, $\varepsilon_\ell{\downarrow}$). 
Tab.~\ref{tab:ablation} reports results on six tasks spanning all three query families; each configuration requires retraining both the tokenizer and the decoder.
The full system achieves the best Recall@1 on every task, confirming that the three components are complementary. 
Two consistent patterns emerge. First, the dominant single-component gain varies by task type, in line with the survival bound. 
Vocabulary scheduling is most effective when the teacher margin is the primary bottleneck, whereas score fusion yields larger gains when the geometric proxy is more reliable. 
Second, prefix distillation enhances the effectiveness of vocabulary scheduling. 
When ranking divergence is large, increasing the margin alone does not reliably preserve the target prefix; once ranking divergence is reduced, vocabulary scheduling becomes consistently beneficial, and the two components exhibit super-additive effects. 
App.~\ref{app:full_ablation} provides full configurations and pairwise interaction analyses. 
Additional encoder sensitivity experiments (App.~\ref{app:encoder_ablation}) show that the gap to dense retrieval persists for vanilla GENIUS across three multimodal encoders, while PRO consistently narrows this gap. 
This suggests that the improvements are not tied to a specific encoder, but arise from mitigating the RQ tokenization and decoding bottleneck.

\begin{figure}[t]
  \centering
  \begin{subfigure}[t]{0.315\textwidth}
    \centering
    \includegraphics[width=\linewidth]{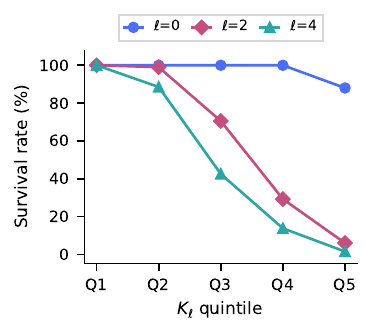}
  \vspace*{-7mm}
    \caption{Ranking divergence}
    \label{fig:theory_validation_kl}
  \end{subfigure}\hfill
  \begin{subfigure}[t]{0.315\textwidth}
    \centering
    \includegraphics[width=\linewidth]{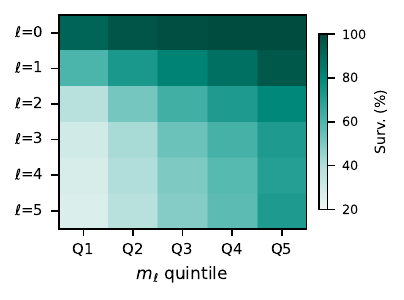}
  \vspace*{-7mm}
    \caption{Teacher margin}
    \label{fig:theory_validation_margin}
  \end{subfigure}\hfill
  \begin{subfigure}[t]{0.315\textwidth}
    \centering
    \includegraphics[width=\linewidth]{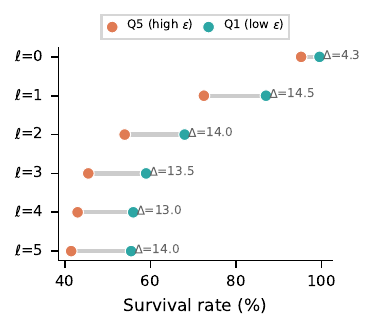}
  \vspace*{-7mm}
    \caption{Decoder mismatch}
    \label{fig:theory_validation_mismatch}
  \end{subfigure}
  \vspace*{-2mm}
  \caption{Survival trends predicted by Thm.~\ref{thm:survival}:
    ranking divergence quintile (Spearman $\rho {=} {-}0.59$), teacher margin
    quintile, and survival gap between the lowest and highest decoder
    mismatch quintiles.
  }
  \label{fig:theory_validation}
\vspace*{-.5\baselineskip}
\end{figure}

\textbf{Diagnostic analysis.}
\phantomsection
\label{sec:diagnostic}
Next, we examine whether the quantities in the survival bound are predictive of prefix survival in practice.
Following Thm.~\ref{thm:survival}, we bin 5{,}000 COCO text-to-image test queries into quintiles of ranking divergence~$K_\ell$, teacher margin~$m_\ell$, and decoder mismatch~$\varepsilon_\ell$, and measure cumulative prefix survival at beam width 20.
Fig.~\ref{fig:theory_validation} shows clear monotonic trends for all three quantities.
Ranking divergence is the strongest predictor of survival, teacher margin exhibits a consistent positive association, and decoder mismatch shows a consistent negative association across levels.
In contrast, the representative reconstruction objective on complete code sequences examined in Sec.~\ref{sec:evidence} shows negligible association with survival, consistent with Thm.~\ref{thm:qdg}, which states that reconstruction quality on complete code sequences does not imply prefix survival.
Complementarily, Fig.~\ref{fig:efficiency}\subref{fig:oracle_prefix_survival} shows that our method improves the per-level survival rate of the target identifier token under the same top-$B$ criterion.
Although the bound derived via Pinsker's inequality is loose, the observed trends align with its predicted dependencies and support the practical relevance of the three quantities.
App.~\ref{app:cross_dataset} extends this analysis to five datasets across all three query families, confirming that the trends for $K_\ell$ and $m_\ell$ generalize, while $\varepsilon_\ell$ remains a broad bottleneck.

\textbf{Generalization Beyond Multimodal Retrieval.}
Since both our analysis and PRO are built on RQ identifiers, we further evaluate generalization on TIGER~\citep{rajput2023recommender}, a representative RQ-based generative recommendation method.
We apply PRO to TIGER and evaluate on the Amazon Beauty, Toys, and Sports categories~\citep{hou2024bridging}.
As shown in App.~\ref{app:gen_rec}, PRO improves all 12 Recall and NDCG metrics over TIGER.
These results suggest that our analysis of prefix survival and PRO extend beyond multimodal retrieval to other RQ-based generative settings.

\section{Related Work}
\label{sec:related}

\textbf{Generative retrieval.}
\ac{GR} formulates retrieval as sequence generation, where an encoder-decoder model maps queries directly to target identifiers~\citep{decao2021autoregressive,tay2022transformer}.
The paradigm has been studied for document retrieval in DSI and NCI~\citep{tay2022transformer,wang2022neural}, extended to recommendation in TIGER~\citep{rajput2023recommender}, and generalized to multimodal retrieval by IRGen~\citep{zhang2024irgen}, GRACE~\citep{li2024grace}, and GENIUS~\citep{genius2025}.
Most related to our work, RIPOR~\citep{zeng2024ripor} improves scalable generative retrieval through prefix-oriented ranking optimization and relevance-based identifier construction.
However, it does not optimize RQ tokenization for target prefix survival under beam search.
Our analysis focuses one stage earlier: Thm.~\ref{thm:qdg} shows that the index itself can deterministically prune the target prefix even with a perfect decoder.
Accordingly, two of our three mechanisms target tokenizer training.
The two approaches address complementary bottlenecks: RIPOR reduces decoder mismatch, while our tokenizer mechanisms reduce ranking divergence and enlarge the teacher margin.

\textbf{Document identifier design.}
The identifier is a key factor of GR.
Early work studied atomic, semantic, and hierarchical document identifiers in DSI and NCI~\citep{tay2022transformer,wang2022neural}.
Later proposals include multiview identifiers~\citep{li2023minder}, lexical index learning~\citep{lee2023glen}, learned tokenization~\citep{sun2023genret}, bottleneck-minimal indexing~\citep{du2024bottleneck}, and identifiers with graded relevance~\citep{tang2024gr2}.
These methods improve identifier semantics, compactness, or relevance
alignment.
Our work addresses a different axis: even with fixed identifiers, the
reconstruction-trained index might fail to preserve the target prefix
under beam search (Thm.~\ref{thm:qdg}).
The mechanisms we propose are therefore complementary to better
identifier designs, because any identifier scheme paired with RQ indexing
can still face this gap.

\textbf{Training and inference mismatch.}
Teacher forcing was introduced for recurrent neural networks~\citep{1998Ateacher}, and later work studied exposure bias caused by the mismatch between training with teacher forcing and free running decoding~\citep{ranzato2015sequence} and the gap between training objectives and beam search behavior~\citep{wiseman2016sequence}.
Inference time score combination, e.g., shallow fusion in neural machine translation, has been used to adapt decoding without retraining~\citep{gulcehre2015using}.
Within \ac{GR}, reinforcement learning from relevance feedback optimizes retrieval behavior~\citep{zhou2023enhancing}, prefix alignment optimization addresses discrepancy between training and inference~\citep{yu2026apao}, and \ac{RQ} refinements improve code utilization and index quality~\citep{kuai2024hourglass}.
Our setting differs in a key respect: this gap arises from a mismatch between index construction and beam decoding, rather than from train and test discrepancy inside the decoder.
The survival bound (Thm.~\ref{thm:survival}) formalizes this coupling by decomposing prefix survival into quantities spanning both stages.

\section{Conclusion}
\label{sec:conclusion}

This paper identifies and formalizes the indexing--decoding gap in \ac{MGR}, showing that identifier learning objectives are insufficient to determine prefix survival under beam decoding.
We derive a survival bound that relates prefix survival to three controllable factors across indexing and decoding, and propose \emph{PRO}, a unified framework that improves prefix survival via 
\begin{enumerate*}[label=(\roman*)]
    \item prefix ranking distillation to align prefix rankings,
    \item vocabulary scheduling to reduce early prefix competition, and 
    \item geometric score fusion to incorporate query-aware similarity into beam search scoring.
\end{enumerate*}
Experiments on $9$ multimodal retrieval tasks show that PRO consistently outperforms state-of-the-art \ac{MGR} methods and improves prefix survival, with strong generalizability across multimodal encoders and RQ-based methods.

\textbf{Limitations.}
While our results consistently improve prefix survival and multimodal retrieval performance, our analysis is specific to RQ-based pipelines and relies on their residual code structure, so its applicability to other identifier schemes remains to be established.
In addition, our framework assumes a two-stage setup with a fixed encoder and tokenizer, which may limit the potential of joint optimization.
Future work includes extending the analysis to other identifier families and developing end-to-end training with survival-aware objectives.

\bibliographystyle{plainnat}
\bibliography{references}

@article{equitz1991successive,
  title   = {Successive Refinement of Information},
  author  = {Equitz, William H. R. and Cover, Thomas M.},
  journal = {IEEE Transactions on Information Theory},
  volume  = {37},
  number  = {2},
  pages   = {269--275},
  year    = {1991},
}

@book{cover2006elements,
  title     = {Elements of Information Theory},
  author    = {Cover, Thomas M. and Thomas, Joy A.},
  publisher = {Wiley},
  year      = {2001},
}

@article{fedotov2003refinements,
  title   = {Refinements of {Pinsker}'s Inequality},
  author  = {Fedotov, Alexei A. and Harremo{\"e}s, Peter and Tops{\o}e, Flemming},
  journal = {IEEE Transactions on Information Theory},
  volume  = {49},
  number  = {6},
  pages   = {1491--1498},
  year    = {2003},
}

@inproceedings{rioul2024historical,
  title     = {A Historical Perspective on {Sch{\"u}tzenberger}-{Pinsker} Inequalities},
  author    = {Rioul, Olivier},
  booktitle = {Geometric Science of Information (GSI)},
  series    = {Lecture Notes in Computer Science},
  pages     = {291--306},
  publisher = {Springer},
  year      = {2023},
}

@article{sason2016f,
  title   = {$f$-Divergence Inequalities},
  author  = {Sason, Igal and Verd{\'u}, Sergio},
  journal = {IEEE Transactions on Information Theory},
  volume  = {62},
  number  = {11},
  pages   = {5973--6006},
  year    = {2016},
}

@article{1998Ateacher,
  title   = {A Learning Algorithm for Continually Running Fully Recurrent Neural Networks},
  author  = {Williams, Ronald J. and Zipser, David},
  journal = {Neural Computation},
  volume  = {1},
  number  = {2},
  pages   = {270--280},
  year    = {1989},
}

@inproceedings{ranzato2015sequence,
  title     = {Sequence Level Training with Recurrent Neural Networks},
  author    = {Ranzato, Marc'Aurelio and Chopra, Sumit and Auli, Michael and Zaremba, Wojciech},
  booktitle = {International Conference on Learning Representations (ICLR)},
  year      = {2016},
}

@inproceedings{wiseman2016sequence,
  title     = {Sequence-to-Sequence Learning as Beam-Search Optimization},
  author    = {Wiseman, Sam and Rush, Alexander M.},
  booktitle = {Proceedings of the Conference on Empirical Methods in Natural Language Processing (EMNLP)},
  pages     = {1296--1306},
  publisher = {Association for Computational Linguistics},
  year      = {2016},
}

@article{gulcehre2015using,
  title      = {On Using Monolingual Corpora in Neural Machine Translation},
  author     = {G{\"u}l{\c{c}}ehre, {\c{C}}a{\u{g}}lar and Firat, Orhan and Xu, Kelvin and Cho, Kyunghyun and Barrault, Lo{\"i}c and Lin, Huei-Chi and Bougares, Fethi and Schwenk, Holger and Bengio, Yoshua},
  journal    = {CoRR},
  volume     = {abs/1503.03535},
  year       = {2015},
  eprinttype = {arXiv},
  eprint     = {1503.03535},
}

@inproceedings{khosla2020supervised,
  title     = {Supervised Contrastive Learning},
  author    = {Khosla, Prannay and Teterwak, Piotr and Wang, Chen and Sarna, Aaron and Tian, Yonglong and Isola, Phillip and Maschinot, Aaron and Liu, Ce and Krishnan, Dilip},
  booktitle = {Advances in Neural Information Processing Systems (NeurIPS)},
  volume    = {33},
  pages     = {18661--18673},
  year      = {2020},
}

@article{hinton2015distilling,
  title      = {Distilling the Knowledge in a Neural Network},
  author     = {Hinton, Geoffrey E. and Vinyals, Oriol and Dean, Jeffrey},
  journal    = {CoRR},
  volume     = {abs/1503.02531},
  year       = {2015},
  eprinttype = {arXiv},
  eprint     = {1503.02531},
}

@article{raffel2020exploring,
  title   = {Exploring the Limits of Transfer Learning with a Unified Text-to-Text Transformer},
  author  = {Raffel, Colin and Shazeer, Noam and Roberts, Adam and Lee, Katherine and Narang, Sharan and Matena, Michael and Zhou, Yanqi and Li, Wei and Liu, Peter J.},
  journal = {Journal of Machine Learning Research},
  volume  = {21},
  pages   = {140:1--140:67},
  year    = {2020},
}

@inproceedings{cao2007learning,
  title     = {Learning to Rank: From Pairwise Approach to Listwise Approach},
  author    = {Cao, Zhe and Qin, Tao and Liu, Tie-Yan and Tsai, Ming-Feng and Li, Hang},
  booktitle = {Proceedings of the International Conference on Machine Learning (ICML)},
  series    = {ACM International Conference Proceeding Series},
  pages     = {129--136},
  publisher = {ACM},
  year      = {2007},
}

@inproceedings{lee2022autoregressive,
  title     = {Autoregressive Image Generation Using Residual Quantization},
  author    = {Lee, Doyup and Kim, Chiheon and Kim, Saehoon and Cho, Minsu and Han, Wook-Shin},
  booktitle = {Proceedings of the IEEE/CVF Conference on Computer Vision and Pattern Recognition (CVPR)},
  pages     = {11513--11522},
  publisher = {IEEE},
  year      = {2022},
}

@article{jegou2011product,
  title   = {Product Quantization for Nearest Neighbor Search},
  author  = {J{\'e}gou, Herv{\'e} and Douze, Matthijs and Schmid, Cordelia},
  journal = {IEEE Transactions on Pattern Analysis and Machine Intelligence},
  volume  = {33},
  number  = {1},
  pages   = {117--128},
  year    = {2011},
}

@inproceedings{lin2014microsoft,
  title     = {Microsoft {COCO}: Common Objects in Context},
  author    = {Lin, Tsung-Yi and Maire, Michael and Belongie, Serge J. and Hays, James and Perona, Pietro and Ramanan, Deva and Doll{\'a}r, Piotr and Zitnick, C. Lawrence},
  booktitle = {European Conference on Computer Vision (ECCV)},
  series    = {Lecture Notes in Computer Science},
  pages     = {740--755},
  publisher = {Springer},
  year      = {2014},
}

@inproceedings{wei2023uniir,
  title     = {{UniIR}: Training and Benchmarking Universal Multimodal Information Retrievers},
  author    = {Wei, Cong and Chen, Yang and Chen, Haonan and Hu, Hexiang and Zhang, Ge and Fu, Jie and Ritter, Alan and Chen, Wenhu},
  booktitle = {European Conference on Computer Vision (ECCV)},
  series    = {Lecture Notes in Computer Science},
  pages     = {387--404},
  publisher = {Springer},
  year      = {2024},
}

@article{young2014flickr30k,
  title   = {From Image Descriptions to Visual Denotations: New Similarity Metrics for Semantic Inference over Event Descriptions},
  author  = {Young, Peter and Lai, Alice and Hodosh, Micah and Hockenmaier, Julia},
  journal = {Transactions of the Association for Computational Linguistics},
  volume  = {2},
  pages   = {67--78},
  year    = {2014},
}

@inproceedings{chang2022webqa,
  title     = {{WebQA}: Multihop and Multimodal {QA}},
  author    = {Chang, Yingshan and Cao, Guihong and Narang, Mridu and Gao, Jianfeng and Suzuki, Hisami and Bisk, Yonatan},
  booktitle = {Proceedings of the IEEE/CVF Conference on Computer Vision and Pattern Recognition (CVPR)},
  pages     = {16474--16483},
  publisher = {IEEE},
  year      = {2022},
}

@inproceedings{fu2023dreamsim,
  title     = {{DreamSim}: Learning New Dimensions of Human Visual Similarity Using Synthetic Data},
  author    = {Fu, Stephanie and Tamir, Netanel and Sundaram, Shobhita and Chai, Lucy and Zhang, Richard and Dekel, Tali and Isola, Phillip},
  booktitle = {Advances in Neural Information Processing Systems (NeurIPS)},
  volume    = {36},
  year      = {2023},
}

@inproceedings{liu2021cirr,
  title     = {Image Retrieval on Real-Life Images with Pre-Trained Vision-and-Language Models},
  author    = {Liu, Zheyuan and Rodriguez Opazo, Cristian and Teney, Damien and Gould, Stephen},
  booktitle = {Proceedings of the IEEE/CVF International Conference on Computer Vision (ICCV)},
  pages     = {2105--2114},
  publisher = {IEEE},
  year      = {2021},
}

@inproceedings{hu2023oven,
  title     = {Open-Domain Visual Entity Recognition: Towards Recognizing Millions of {Wikipedia} Entities},
  author    = {Hu, Hexiang and Luan, Yi and Chen, Yang and Khandelwal, Urvashi and Joshi, Mandar and Lee, Kenton and Toutanova, Kristina and Chang, Ming-Wei},
  booktitle = {Proceedings of the IEEE/CVF International Conference on Computer Vision (ICCV)},
  pages     = {12031--12041},
  publisher = {IEEE},
  year      = {2023},
}

@inproceedings{chen2023infoseek,
  title     = {Can Pre-Trained Vision and Language Models Answer Visual Information-Seeking Questions?},
  author    = {Chen, Yang and Hu, Hexiang and Luan, Yi and Sun, Haitian and Changpinyo, Soravit and Ritter, Alan and Chang, Ming-Wei},
  booktitle = {Proceedings of the Conference on Empirical Methods in Natural Language Processing (EMNLP)},
  pages     = {14948--14968},
  publisher = {Association for Computational Linguistics},
  year      = {2023},
}

@article{hou2024bridging,
  title      = {Bridging Language and Items for Retrieval and Recommendation},
  author     = {Hou, Yupeng and Li, Jiacheng and He, Zhankui and Yan, An and Chen, Xiusi and McAuley, Julian J.},
  journal    = {CoRR},
  volume     = {abs/2403.03952},
  year       = {2024},
  eprinttype = {arXiv},
  eprint     = {2403.03952},
}

@inproceedings{radford2021learning,
  title     = {Learning Transferable Visual Models from Natural Language Supervision},
  author    = {Radford, Alec and Kim, Jong Wook and Hallacy, Chris and Ramesh, Aditya and Goh, Gabriel and Agarwal, Sandhini and Sastry, Girish and Askell, Amanda and Mishkin, Pamela and Clark, Jack and Krueger, Gretchen and Sutskever, Ilya},
  booktitle = {Proceedings of the International Conference on Machine Learning (ICML)},
  series    = {Proceedings of Machine Learning Research},
  pages     = {8748--8763},
  publisher = {PMLR},
  year      = {2021},
}

@inproceedings{huang2026llm2clip,
  title     = {{LLM2CLIP}: Powerful Language Model Unlocks Richer Cross-Modality Representation},
  author    = {Huang, Weiquan and Wu, Aoqi and Yang, Yifan and Luo, Xufang and Yang, Yuqing and Naseem, Usman and Wang, Chunyu and Dai, Qi and Dai, Xiyang and Chen, Dongdong and Luo, Chong and Qiu, Lili and Hu, Liang},
  booktitle = {Proceedings of the AAAI Conference on Artificial Intelligence (AAAI)},
  pages     = {5131--5139},
  publisher = {AAAI Press},
  year      = {2026},
}

@inproceedings{decao2021autoregressive,
  title     = {Autoregressive Entity Retrieval},
  author    = {De Cao, Nicola and Izacard, Gautier and Riedel, Sebastian and Petroni, Fabio},
  booktitle = {International Conference on Learning Representations (ICLR)},
  publisher = {OpenReview.net},
  year      = {2021},
}

@inproceedings{tay2022transformer,
  title     = {Transformer Memory as a Differentiable Search Index},
  author    = {Tay, Yi and Tran, Vinh and Dehghani, Mostafa and Ni, Jianmo and Bahri, Dara and Mehta, Harsh and Qin, Zhen and Hui, Kai and Zhao, Zhe and Gupta, Jai Prakash and Schuster, Tal and Cohen, William W. and Metzler, Donald},
  booktitle = {Advances in Neural Information Processing Systems (NeurIPS)},
  volume    = {35},
  year      = {2022},
}

@inproceedings{wang2022neural,
  title     = {A Neural Corpus Indexer for Document Retrieval},
  author    = {Wang, Yujing and Hou, Yingyan and Wang, Haonan and Miao, Ziming and Wu, Shibin and Chen, Qi and Xia, Yuqing and Chi, Chengmin and Zhao, Guoshuai and Liu, Zheng and Xie, Xing and Sun, Hao and Deng, Weiwei and Zhang, Qi and Yang, Mao},
  booktitle = {Advances in Neural Information Processing Systems (NeurIPS)},
  volume    = {35},
  year      = {2022},
}

@inproceedings{rajput2023recommender,
  title     = {Recommender Systems with Generative Retrieval},
  author    = {Rajput, Shashank and Mehta, Nikhil and Singh, Anima and Keshavan, Raghunandan Hulikal and Vu, Trung and Heldt, Lukasz and Hong, Lichan and Tay, Yi and Tran, Vinh Q. and Samost, Jonah and Kula, Maciej and Chi, Ed H. and Sathiamoorthy, Mahesh},
  booktitle = {Advances in Neural Information Processing Systems (NeurIPS)},
  volume    = {36},
  year      = {2023},
}

@inproceedings{li2023minder,
  title     = {{MINDER}: Multiview Identifiers Enhanced Generative Retrieval},
  author    = {Li, Yongqi and Yang, Nan and Wang, Liang and Wei, Furu and Li, Wenjie},
  booktitle = {Proceedings of the Annual Meeting of the Association for Computational Linguistics (ACL)},
  pages     = {6636--6648},
  publisher = {Association for Computational Linguistics},
  year      = {2023},
}

@inproceedings{sun2023genret,
  title     = {Learning to Tokenize for Generative Retrieval},
  author    = {Sun, Weiwei and Yan, Lingyong and Chen, Zheng and Wang, Shuaiqiang and Zhu, Haichao and Ren, Pengjie and Chen, Zhumin and Yin, Dawei and de Rijke, Maarten and Ren, Zhaochun},
  booktitle = {Advances in Neural Information Processing Systems (NeurIPS)},
  volume    = {36},
  year      = {2023},
}

@inproceedings{lee2023glen,
  title     = {{GLEN}: Generative Retrieval via Lexical Index Learning},
  author    = {Lee, Sunkyung and Choi, Minjin and Lee, Jongwuk},
  booktitle = {Proceedings of the Conference on Empirical Methods in Natural Language Processing (EMNLP)},
  pages     = {7693--7704},
  publisher = {Association for Computational Linguistics},
  year      = {2023},
}

@inproceedings{zhou2023enhancing,
  title     = {Enhancing Generative Retrieval with Reinforcement Learning from Relevance Feedback},
  author    = {Zhou, Yujia and Dou, Zhicheng and Wen, Ji-Rong},
  booktitle = {Proceedings of the Conference on Empirical Methods in Natural Language Processing (EMNLP)},
  pages     = {12481--12490},
  publisher = {Association for Computational Linguistics},
  year      = {2023},
}

@inproceedings{zeng2024ripor,
  title     = {Scalable and Effective Generative Information Retrieval},
  author    = {Zeng, Hansi and Luo, Chen and Jin, Bowen and Sarwar, Sheikh Muhammad and Wei, Tianxin and Zamani, Hamed},
  booktitle = {Proceedings of the ACM Web Conference (WWW)},
  pages     = {1441--1452},
  publisher = {ACM},
  year      = {2024},
}

@inproceedings{du2024bottleneck,
  title     = {Bottleneck-Minimal Indexing for Generative Document Retrieval},
  author    = {Du, Xin and Xiu, Lixin and Tanaka-Ishii, Kumiko},
  booktitle = {Proceedings of the International Conference on Machine Learning (ICML)},
  series    = {Proceedings of Machine Learning Research},
  pages     = {11888--11904},
  publisher = {PMLR / OpenReview.net},
  year      = {2024},
}

@inproceedings{tang2024gr2,
  title     = {Generative Retrieval Meets Multi-Graded Relevance},
  author    = {Tang, Yubao and Zhang, Ruqing and Guo, Jiafeng and de Rijke, Maarten and Chen, Wei and Cheng, Xueqi},
  booktitle = {Advances in Neural Information Processing Systems (NeurIPS)},
  volume    = {38},
  year      = {2024},
}

@inproceedings{kuai2024hourglass,
  title     = {Breaking the Hourglass Phenomenon of Residual Quantization: Enhancing the Upper Bound of Generative Retrieval},
  author    = {Kuai, Zhirui and Chen, Zuxu and Wang, Huimu and Li, Mingming and Miao, Dadong and Wang, Binbin and Chen, Xusong and Kuang, Li and Han, Yuxing and Wang, Jiaxing and Tang, Guoyu and Liu, Lin and Wang, Songlin and Zhuo, Jingwei},
  booktitle = {Proceedings of the Conference on Empirical Methods in Natural Language Processing: Industry Track (EMNLP)},
  pages     = {677--685},
  publisher = {Association for Computational Linguistics},
  year      = {2024},
}

@article{yu2026apao,
  title      = {{APAO}: Adaptive Prefix-Aware Optimization for Generative Recommendation},
  author     = {Yu, Yuanqing and Wang, Yifan and Ma, Weizhi and Guo, Zhiqiang and Zhang, Min},
  journal    = {CoRR},
  volume     = {abs/2603.02730},
  year       = {2026},
  eprinttype = {arXiv},
  eprint     = {2603.02730},
}

@inproceedings{zhang2024irgen,
  title     = {{IRGen}: Generative Modeling for Image Retrieval},
  author    = {Zhang, Yidan and Zhang, Ting and Chen, Dong and Wang, Yujing and Chen, Qi and Xie, Xing and Sun, Hao and Deng, Weiwei and Zhang, Qi and Yang, Fan and Yang, Mao and Liao, Qingmin and Wang, Jingdong and Guo, Baining},
  booktitle = {European Conference on Computer Vision (ECCV)},
  series    = {Lecture Notes in Computer Science},
  pages     = {21--41},
  publisher = {Springer},
  year      = {2024},
}

@inproceedings{li2024grace,
  title     = {Generative Cross-Modal Retrieval: Memorizing Images in Multimodal Language Models for Retrieval and Beyond},
  author    = {Li, Yongqi and Wang, Wenjie and Qu, Leigang and Nie, Liqiang and Li, Wenjie and Chua, Tat-Seng},
  booktitle = {Proceedings of the Annual Meeting of the Association for Computational Linguistics (ACL)},
  pages     = {11851--11861},
  publisher = {Association for Computational Linguistics},
  year      = {2024},
}

@inproceedings{genius2025,
  title     = {{GENIUS}: A Generative Framework for Universal Multimodal Search},
  author    = {Kim, Sungyeon and Zhu, Xinliang and Lin, Xiaofan and Bastan, Muhammet and Gray, Douglas and Kwak, Suha},
  booktitle = {Proceedings of the IEEE/CVF Conference on Computer Vision and Pattern Recognition (CVPR)},
  pages     = {19659--19669},
  publisher = {Computer Vision Foundation / IEEE},
  year      = {2025},
}

@inproceedings{cai2025avg,
  title     = {Revolutionizing Text-to-Image Retrieval as Autoregressive Token-to-Voken Generation},
  author    = {Li, Yongqi and Cai, Hongru and Wang, Wenjie and Qu, Leigang and Wei, Yinwei and Li, Wenjie and Nie, Liqiang and Chua, Tat-Seng},
  booktitle = {Proceedings of the International ACM SIGIR Conference on Research and Development in Information Retrieval (SIGIR)},
  pages     = {813--822},
  publisher = {ACM},
  year      = {2025},
}

@article{wu2025semcore,
  title      = {{SemCORE}: A Semantic-Enhanced Generative Cross-Modal Retrieval Framework with {MLLMs}},
  author     = {Li, Haoxuan and Bin, Yi and Ma, Yunshan and Wang, Guoyin and Yang, Yang and Ng, See-Kiong and Chua, Tat-Seng},
  journal    = {CoRR},
  volume     = {abs/2504.13172},
  year       = {2025},
  eprinttype = {arXiv},
  eprint     = {2504.13172},
}

\newpage
\appendix

\setcounter{theorem}{0}
\renewcommand{\thetheorem}{\Alph{section}.\arabic{theorem}}
\renewcommand{\theHtheorem}{\Alph{section}.\arabic{theorem}}  
\renewcommand{\qedsymbol}{}

\section*{Supplementary Material}

Our supplementary material has the following sections:
\begin{itemize}
    \item Proof details
    \item Additional experimental results
    \item Implementation details
\end{itemize}

\section{Proof Details}
\label{app:proofs}

We collect full proofs of the theoretical results stated in the main
text.  All notation follows Secs.~\ref{sec:qdg}--\ref{sec:method};
we briefly recall the key symbols for convenience.

\paragraph{Notation.}
We use the following groups of symbols throughout the proofs.
\begin{itemize}[leftmargin=1.8em,itemsep=0.15em]
  \item \textbf{Embeddings and codes.}
  $x_j \in \RR^d$ is the dense embedding of target item~$j$.
  $\hat{x}_{j,\ell} := \sum_{t=1}^{\ell} c_{t,z_j^t}$ is its partial
  reconstruction after $\ell$ RQ levels, and
  $\hat{x}_j := \hat{x}_{j,L}$ is its full reconstruction.
  $z_j^{1:\ell}$ denotes the length-$\ell$ code prefix.
  \item \textbf{Beam event.}
  $p^+_\ell$ is the ground-truth prefix at level~$\ell$, and
  $S_\ell$ is the event that $p^+_\ell$ survives in the beam.
  \item \textbf{Prefix distributions.}
  $\pi^*_\ell$ scores prefixes using pre-quantization embedding
  similarities; $\bar{\pi}_\ell$ is the quantized oracle distribution;
  and $\tilde{\pi}_\ell$ is the decoder prefix distribution.
  These objects are probability mass functions (PMFs) over the finite prefix
  set at level~$\ell$.  A lower-case symbol such as~$p$ denotes a
  prefix value; when a random prefix is needed, we use an upper-case
  symbol and write, for example, $P \sim \pi^*_\ell(\cdot)$.  Here,
  ``teacher'' refers to $\pi^*_\ell$.
  \item \textbf{Survival factors.}
  $K_\ell := \KL(\pi^*_\ell \| \bar{\pi}_\ell)$,
  $m_\ell := \pi^*_\ell(p^+_\ell) - \pi^*_{\ell,B+1}$, and
  $\varepsilon_\ell := \TV(\bar{\pi}_\ell, \tilde{\pi}_\ell)$.
  The KL and TV terms are taken between PMFs over the same prefix set.
  Here, $\pi^*_{\ell,B+1}$ denotes the probability of the
  $(B{+}1)$-th ranked prefix under the teacher.
\end{itemize}

For completeness, for each prefix value~$p$, the $3$ prefix PMFs are
\[
  \pi^*_\ell(p) \propto \exp\!\bigl(s^*_\ell(p)/\tau\bigr), \qquad
  \bar{\pi}_\ell(p) \propto \exp\!\bigl(\bar{s}_\ell(p)/\tau\bigr),
  \qquad
  \tilde{\pi}_\ell(p) \propto \exp\!\bigl(\tilde{s}_\ell(p)/\tau\bigr),
\]
where
\[
  s^*_\ell(p) = \max_{j:\,z_j^{1:\ell}=p} \langle q, x_j \rangle
\]
is the frozen-encoder similarity score,
\[
  \bar{s}_\ell(p) = \langle q, \hat{x}_{j,\ell} \rangle
\]
is the partial-reconstruction score shared by all target items with
prefix~$p$, and $\tilde{s}_\ell(p)$ denotes the actual decoder score.

\subsection{Proof of Theorem~\ref{thm:qdg}}
\label{app:proof_qdg}

\paragraph{Proof idea.}
The theorem separates existing tokenizer objectives from prefix survival
during beam search.
We construct two RQ tokenizers over the same query and item embeddings.
They select the same code indices and have the same complete query and
item reconstructions; their residual-fitting costs also match by
construction. Therefore, the contrastive, residual-fitting, and
reconstruction objectives considered in the theorem take identical values.
The difference lies only in how the first two selected codewords decompose
the same complete reconstruction.
In $\mathcal{Q}^+$, the target prefix has the highest quantized-oracle
score at every level.
In $\mathcal{Q}^-$, the first selected target codeword is shifted away
from the query direction and the second selected target codeword is shifted
back, leaving the complete reconstruction and objective values unchanged.
This shift lowers the target first-level prefix score below $B$ competing
prefixes, so beam search prunes it.
Thus, existing tokenizer objectives do not by themselves determine prefix
survival.
\paragraph{Construction.}
Let $K := B + 1$.  We construct the ambient space $\RR^d$ with
$d = 2K + L + 1$ as a direct sum of mutually orthogonal subspaces:
\[
  \RR^d \;=\; S \;\oplus\; U \;\oplus\; A_1 \;\oplus\; A_2 \;\oplus\;
  A_3 \;\oplus\; \cdots \;\oplus\; A_L \;\oplus\; R,
\]
where $\dim S = \dim U = 1$, $\dim A_1 = \dim A_2 = K$,
$\dim A_t = 1$ for $t = 3, \ldots, L$, and $\dim R = 1$.

Let $s \in S$ and $u \in U$ be unit vectors, and fix a large constant
$H > 0$.
The query embedding is $q := s + H u$.
In~$A_1$, fix $K$ orthonormal vectors
$a_1^{(1)}, \ldots, a_K^{(1)}$; in~$A_2$, fix $K$ orthonormal vectors
$a_1^{(2)}, \ldots, a_K^{(2)}$.
For $t \geq 3$, let $a_\star^{(t)}$ be a unit vector spanning~$A_t$.
Let $r$ be a unit vector spanning~$R$.
Fix a constant $M > 2$.

\textit{Codebooks for $\mathcal{Q}^+$.}\;
Level~1 ($K{+}1$ codewords):
$c_{1,0}^+ = q$;\;
$c_{1,1}^+ = s + Ma_1^{(1)}$;\;
$c_{1,j}^+ = Ma_j^{(1)}$ for $j = 2, \ldots, K$.
Level~2 ($K{+}1$ codewords):
$c_{2,0}^+ = 0$;\;
$c_{2,j}^+ = Ma_j^{(2)}$ for $j = 1, \ldots, K$.
Level $t \geq 3$ (two codewords each):
$c_{t,0}^+ = 0$;\;
$c_{t,\star}^+ = Ma_\star^{(t)}$.
The zero codewords are selected only by the query. Unused
dummy codewords can be placed sufficiently far away without changing any
assignment, objective value, or beam-search outcome below.
Thus the construction only requires finite vocabularies with
$V_1,V_2 \geq K+1$ and $V_t \geq 2$ for $t \geq 3$.

\textit{Items.}\;
$K = B + 1$ target items:
\begin{equation}
  x_j \;=\; c_{1,j}^+ + c_{2,j}^+
         + \sum_{t=3}^{L} c_{t,\star}^+ + \alpha\, r,
  \qquad j = 1, \ldots, K.
\label{eq:app_doc}
\end{equation}
Target item $j = 1$ is the target.

\begin{lemma}[Greedy assignment and reconstruction optimality under $\mathcal{Q}^+$]
\label{lem:app_mse}
Under $\mathcal{Q}^+$, the greedy nearest-neighbor RQ assigns
the query the code sequence $(0,0,\ldots,0)$ and target item~$x_j$ the
code sequence $(j,\, j,\, \star, \ldots, \star)$.
The resulting target reconstruction error is $\alpha^2$ for every item,
which is the global minimum over all code assignments for the constructed
codebooks.
\end{lemma}

\begin{proof}
\textit{Query.}\;
The query $q$ exactly matches $c_{1,0}^+$ at level~1 and has zero
residual afterwards, so it selects $(0,0,\ldots,0)$.

\textit{Level~1.}\;
The distance from~$x_j$ to the correct codeword~$c_{1,j}^+$ is
$\|x_j - c_{1,j}^+\|^2 = (L{-}1)M^2 + \alpha^2$
(only $A_2, \ldots, A_L, R$ components remain).
For any incorrect token $k \neq j$, the $A_1$-component mismatch
contributes $\|Ma_j^{(1)} - Ma_k^{(1)}\|^2 = 2M^2$ by orthonormality,
so $\|x_j - c_{1,k}^+\|^2 \geq 2M^2 + (L{-}1)M^2 + \alpha^2$.
The query-only token $c_{1,0}^+=q$ incurs an additional $H^2$ and is
therefore not selected by any target item.
The gap is at least $2M^2 > 0$; greedy selects token~$j$.

\textit{Level~2.}\;
After subtracting $c_{1,j}^+$, the residual has zero $A_1$-component.
By the same orthonormality argument, the correct token~$j$ achieves
distance $(L{-}2)M^2 + \alpha^2$.
Any incorrect item token incurs an additional $2M^2$, and the
query-only zero token leaves an additional $M^2$ component.
Greedy selects token~$j$.

\textit{Level $t \geq 3$.}\;
The item codeword $c_{t,\star}^+$ removes the remaining $A_t$ component,
whereas the query-only zero token does not. Greedy selects
$c_{t,\star}^+$.

\textit{Full reconstruction and MSE.}\;
$\hat{q}=q$ and
$\hat{x}_j = c_{1,j}^+ + c_{2,j}^+ + \sum_{t \geq 3} c_{t,\star}^+
= x_j - \alpha\, r$, so the target reconstruction error is $\alpha^2$.

\textit{Global optimality.}\;
Every codeword lies in
$S \oplus U \oplus A_1 \oplus \cdots \oplus A_L$, so any reconstruction
has zero component in~$R$.  Since each $x_j$ has a nonzero
$R$-component $\alpha r$, the Pythagorean theorem gives
$\|x_j - \tilde{x}_j\|^2 \geq \alpha^2$ for any code assignment.
Equality is achieved above.
\end{proof}

\begin{lemma}[Target survival under trie-constrained quantized oracle beam search in $\mathcal{Q}^+$]
\label{lem:app_survival_plus}
Under trie-constrained beam search over the valid item identifiers with
quantized oracle scores, the target prefix
survives at every level $\ell = 1, \ldots, L$.
\end{lemma}

\begin{proof}
For the target ($j = 1$): level~1 contributes
$c_{1,1}^+ = s + Ma_1^{(1)}$, which has
$\langle q, c_{1,1}^+ \rangle = 1$.
Every subsequent level adds a codeword orthogonal to~$q$, so
$\langle q, \hat{x}_{1,\ell}^+ \rangle = 1$ for all~$\ell$.
For any competitor ($j \geq 2$): $c_{1,j}^+ = Ma_j^{(1)} \in A_1$, so
$\langle q, \hat{x}_{j,\ell}^+ \rangle = 0$ for all~$\ell$.
The target has the unique highest score at every level.
\end{proof}

\paragraph{Construction of $\mathcal{Q}^-$.}
Let $\delta := M - \sqrt{M^2 - 4}$, so that
$4 + (M-\delta)^2 = M^2$.
Modify exactly two codewords:
$c_{1,1}^- = -s + Ma_1^{(1)} + \delta a_1^{(2)}$,\;
$c_{2,1}^- = 2s + (M-\delta)a_1^{(2)}$.
All other codewords are unchanged.

\begin{proof}[Proof of Thm.~\ref{thm:qdg}]
We verify three claims.

\textit{Claim~1: Greedy assignment remains valid under
$\mathcal{Q}^-$.}\;
For the target at level~1, the distance to $c_{1,1}^-$ is
$4 + (M{-}\delta)^2 + (L{-}2)M^2 + \alpha^2
= (L{-}1)M^2 + \alpha^2$, while any other token $k \geq 2$ gives
$1 + (L{+}1)M^2 + \alpha^2$.
The query-only token $c_{1,0}^-=q$ incurs an additional $H^2$ and is
not selected by any target item.
Token~1 remains nearest.
At level~2, the residual after $c_{1,1}^-$ is
$2s + (M{-}\delta)a_1^{(2)}
+ \sum_{t \geq 3} c_{t,\star}^+ + \alpha\, r$;
the distance to $c_{2,1}^- = 2s + (M{-}\delta)a_1^{(2)}$ is
$(L{-}2)M^2 + \alpha^2$.
Any incorrect item token incurs an additional
$4 + (M{-}\delta)^2 + M^2$, and the query-only zero token leaves an
additional $4 + (M{-}\delta)^2 = M^2$.
For non-target items ($j \geq 2$), the correct distances are unchanged.
The modified level-1 token has distance
$1+(L{+}1)M^2+\delta^2+\alpha^2$, and the modified level-2 token has
an extra $4+M^2+(M-\delta)^2$ over the correct choice.
Levels $t \geq 3$ are unchanged.
The query still selects $(0,0,\ldots,0)$ because the query-only codewords
are unchanged. Therefore $\mathcal{Q}^-$ assigns the same code sequences as
$\mathcal{Q}^+$.

\textit{Claim~2: Identical standard tokenizer objective values.}\;
Since code assignments are unchanged, the full reconstruction of the
target is
$\hat{x}_1^- = c_{1,1}^- + c_{2,1}^- + \sum_{t \geq 3} c_{t,\star}^-
= (-s + Ma_1^{(1)} + \delta a_1^{(2)})
  + (2s + (M-\delta)a_1^{(2)}) + \cdots
= (s + Ma_1^{(1)}) + Ma_1^{(2)} + \cdots = \hat{x}_1^+$.
For $j \geq 2$, all codewords are unchanged, so $\hat{x}_j^- = \hat{x}_j^+$.
The query reconstruction is also identical because
$c_{1,0}^-=c_{1,0}^+=q$ and $c_{\ell,0}^-=c_{\ell,0}^+=0$ for
$\ell \geq 2$.
Both quantizers therefore produce identical complete query and item
reconstructions. Consequently, objectives evaluated on complete
reconstructions take identical values under $\mathcal{Q}^+$ and
$\mathcal{Q}^-$, including target reconstruction losses and
quantized-space MSE terms such as $\|\hat{q}-\hat{x}_j\|^2$.
The target reconstruction error is $\alpha^2$ under both.
The squared residual-fitting term
$\sum_{\ell}\|r_{\ell-1}-c_{\ell,z^\ell}\|^2$ also matches.
For the query, the fitting cost is zero under both quantizers.
For the target, the level-1 fitting cost under $\mathcal{Q}^-$ is
$4+(M-\delta)^2+(L-2)M^2+\alpha^2=(L-1)M^2+\alpha^2$, identical to
$\mathcal{Q}^+$; the level-2 cost is $(L-2)M^2+\alpha^2$ under both,
and later levels are unchanged.
All non-target fitting costs are unchanged because their selected
codewords are unchanged.
Finally, any contrastive term computed on the pre-quantization embeddings
is identical because the query and item embeddings are shared.
The Pythagorean lower bound from Lem.~\ref{lem:app_mse} applies to
$\mathcal{Q}^-$ as well (all codewords still lie in
$S \oplus U \oplus A_1 \oplus \cdots \oplus A_L$), so $\alpha^2$ remains
the global target-reconstruction optimum.

\textit{Claim~3: Target is pruned under trie-constrained quantized oracle beam search in
$\mathcal{Q}^-$.}\;
The level-1 prefix scores are
$\langle q, c_{1,1}^- \rangle = -1$ and
$\langle q, c_{1,j}^- \rangle = 0$ for $j \geq 2$.
The target is strictly last; with beam width $B = K - 1$, it is pruned
deterministically.

\textit{Impossibility.}\;
The target survives under $\mathcal{Q}^+$ (Lem.~\ref{lem:app_survival_plus})
and is pruned under $\mathcal{Q}^-$ (Claim~3), yet the two quantizers have identical
values for the standard tokenizer-training terms covered in
Claim~2.
Let $\Phi(\mathcal{Q})$ collect these objective values.
Then $\Phi(\mathcal{Q}^+) = \Phi(\mathcal{Q}^-)$, but the target prefix
survives under $\mathcal{Q}^+$ and is pruned under $\mathcal{Q}^-$.
Thus no criterion depending only on $\Phi$ can determine the beam-search
outcome, even under trie-constrained quantized oracle beam search.
\end{proof}

\subsection{Preparatory tools for the survival bound}
\label{app:tv_tools}

We recall the definitions and standard results used in the proof of
Thm.~\ref{thm:survival}.

\begin{definition}[Total variation distance]
\label{def:tv}
For probability distributions $P, Q$ over a finite set~$\mathcal{X}$,
the \emph{total variation distance} is
\begin{equation}
  \TV(P, Q)
    \;:=\; \frac{1}{2} \sum_{x \in \mathcal{X}} |P(x) - Q(x)|
    \;=\; \max_{A \subseteq \mathcal{X}} |P(A) - Q(A)|.
\label{eq:app_tv_def}
\end{equation}
The two expressions are equal by a standard identity.
\end{definition}

\begin{lemma}[Pointwise bound]
\label{lem:app_tv_point}
For any single element $x \in \mathcal{X}$,
\begin{equation}
  |P(x) - Q(x)| \;\leq\; \TV(P, Q).
\label{eq:app_tv_point}
\end{equation}
\end{lemma}

\begin{proof}
By the second expression in Eq.~\ref{eq:app_tv_def},
$\TV(P,Q) = \max_A |P(A) - Q(A)| \geq |P(\{x\}) - Q(\{x\})|$.
\end{proof}

\begin{lemma}[Triangle inequality]
\label{lem:app_tv_tri}
For distributions $P, Q, R$ over the same finite set,
\begin{equation}
  \TV(P, R) \;\leq\; \TV(P, Q) + \TV(Q, R).
\label{eq:app_tv_tri}
\end{equation}
\end{lemma}

\begin{proof}
For each $x$, $|P(x) - R(x)| \leq |P(x) - Q(x)| + |Q(x) - R(x)|$.
Summing over~$x$ and dividing by~2 yields the result.
\end{proof}

\begin{lemma}[Pinsker's inequality {\citep{cover2006elements}}]
\label{lem:app_pinsker}
For distributions $P, Q$ with $P$ absolutely continuous with respect
to~$Q$,
\begin{equation}
  \TV(P, Q) \;\leq\; \sqrt{\frac{1}{2}\,\KL(P \| Q)}.
\label{eq:app_pinsker}
\end{equation}
\end{lemma}

\subsection{Proof of Theorem~\ref{thm:survival} (Survival bound)}
\label{app:proof_survival}

We first establish a per-level result for the quantized oracle alone
(Lem.~\ref{lem:app_kl_topB}), then extend it to the actual decoder
prefix distribution via the triangle inequality to obtain
Thm.~\ref{thm:survival}.

\begin{lemma}[Per-level KL controls top-$B$ under quantized oracle]
\label{lem:app_kl_topB}
Fix a query~$q$ and a level~$\ell$.
If $p^+_\ell$ is in the teacher's top-$B$ under $\pi^*_\ell$
(i.e., $m_\ell > 0$) and
\begin{equation}
  K_\ell \;=\; \KL(\pi^*_\ell \,\|\, \bar{\pi}_\ell)
  \;<\; \frac{m_\ell^2}{2},
\label{eq:app_kl_cond}
\end{equation}
then $p^+_\ell$ is also in the top-$B$ under the quantized
oracle~$\bar{\pi}_\ell$.
\end{lemma}

\begin{proof}
\textit{Step~1: assume $\TV(\pi^*_\ell, \bar{\pi}_\ell) < m_\ell / 2$
and show ranking preservation.}

Let $\delta := \TV(\pi^*_\ell, \bar{\pi}_\ell)$ and suppose
$\delta < m_\ell/2$.
By Lem.~\ref{lem:app_tv_point}, the probability of each individual
prefix can change by at most~$\delta$ when moving from~$\pi^*_\ell$
to~$\bar{\pi}_\ell$.

The ground-truth prefix can \emph{drop} by at most~$\delta$:
\begin{equation}
  \bar{\pi}_\ell(p^+_\ell)
    \;\geq\; \pi^*_\ell(p^+_\ell) - \delta.
\label{eq:app_gt_drop}
\end{equation}

Any prefix~$p$ with teacher rank under $\pi^*_\ell$ greater than $B$ (i.e.,
$\pi^*_\ell(p) \leq \pi^*_{\ell,B+1}$) can \emph{rise} by at
most~$\delta$:
\begin{equation}
  \bar{\pi}_\ell(p)
    \;\leq\; \pi^*_\ell(p) + \delta
    \;\leq\; \pi^*_{\ell,B+1} + \delta.
\label{eq:app_comp_rise}
\end{equation}

Subtracting Eq.~\ref{eq:app_comp_rise} from Eq.~\ref{eq:app_gt_drop}:
\begin{align}
  \bar{\pi}_\ell(p^+_\ell) - \bar{\pi}_\ell(p)
    &\;\geq\; \bigl[\pi^*_\ell(p^+_\ell) - \delta\bigr]
              - \bigl[\pi^*_{\ell,B+1} + \delta\bigr]
    \notag \\
    &\;=\; \underbrace{\pi^*_\ell(p^+_\ell)
           - \pi^*_{\ell,B+1}}_{= \,m_\ell}
           - 2\delta
    \notag \\
    &\;=\; m_\ell - 2\delta
    \;>\; 0,
\label{eq:app_margin_gap}
\end{align}
where the strict inequality uses $\delta < m_\ell/2$.
Since $\bar{\pi}_\ell(p^+_\ell)$ is strictly greater than the
quantized oracle probability of \emph{every} prefix whose teacher rank
exceeds~$B$, and at most $B - 1$ other prefixes can have teacher rank
$\leq B$, we conclude $p^+_\ell \in \topB(\bar{\pi}_\ell)$.

\medskip
\textit{Step~2: use Pinsker's inequality to convert the KL condition to
a TV condition.}

By Pinsker's inequality (Lem.~\ref{lem:app_pinsker}):
\[
  \delta
    = \TV(\pi^*_\ell,\, \bar{\pi}_\ell)
    \leq \sqrt{\frac{K_\ell}{2}}.
\]
For the condition in Step~1 ($\delta < m_\ell/2$) to hold, it suffices
that
\[
  \sqrt{\frac{K_\ell}{2}} < \frac{m_\ell}{2}.
\]
Squaring both sides (both are non-negative):
\[
  \frac{K_\ell}{2} < \frac{m_\ell^2}{4},
  \qquad\text{i.e.,}\qquad
  K_\ell < \frac{m_\ell^2}{2}.
\]
This is precisely the hypothesis in Eq.~\ref{eq:app_kl_cond}.
\end{proof}


We now extend Lem.~\ref{lem:app_kl_topB} from the quantized oracle to
the actual decoder prefix distribution, arriving at Thm.~\ref{thm:survival}.

\begin{proof}[Proof of Thm.~\ref{thm:survival}]
Fix a query~$q$ for which the condition
$\sqrt{K_\ell/2} + \varepsilon_\ell < m_\ell/2$ holds at every
level~$\ell$.  We prove by induction that $p^+_\ell$ remains in the
beam at every level.

\textit{Per-level argument.}\;
Fix any level~$\ell$ and assume $S_{\ell-1}$ holds (the inductive
hypothesis; for $\ell = 1$, there is no previous pruning step).
We need to show
$p^+_\ell \in \topB(\tilde{\pi}_\ell)$ over the candidate set at
level~$\ell$.

By the triangle inequality for TV
(Lem.~\ref{lem:app_tv_tri}):
\begin{equation}
  \TV(\pi^*_\ell,\, \tilde{\pi}_\ell)
  \;\leq\;
    \underbrace{\TV(\pi^*_\ell,\, \bar{\pi}_\ell)}_{\text{quantization shift}}
    \;+\;
    \underbrace{\TV(\bar{\pi}_\ell,\, \tilde{\pi}_\ell)}_{\text{decoder mismatch}}.
\label{eq:app_tri_expand}
\end{equation}
Applying Pinsker (Lem.~\ref{lem:app_pinsker}) to the first term and
using the definition $\varepsilon_\ell :=
\TV(\bar{\pi}_\ell, \tilde{\pi}_\ell)$ for the second:
\begin{equation}
  \TV(\pi^*_\ell,\, \tilde{\pi}_\ell)
  \;\leq\;
    \sqrt{\frac{K_\ell}{2}}
    \;+\; \varepsilon_\ell.
\label{eq:app_combined_tv}
\end{equation}
By the hypothesis of the theorem, the right-hand side is strictly less
than~$m_\ell/2$.  Denoting the left-hand side by
$\delta' := \TV(\pi^*_\ell, \tilde{\pi}_\ell)$, we have
$\delta' < m_\ell/2$.

The remainder of the argument is identical to Step~1 of
Lem.~\ref{lem:app_kl_topB}, but with $\tilde{\pi}_\ell$ in place
of~$\bar{\pi}_\ell$:
by Lem.~\ref{lem:app_tv_point}, the ground-truth prefix can drop by at
most~$\delta'$, and any competitor ranked below $B$ by the teacher can
rise by at most~$\delta'$.  Subtracting (as
in Eq.~\ref{eq:app_margin_gap}):
\[
  \tilde{\pi}_\ell(p^+_\ell) - \tilde{\pi}_\ell(p)
  \;\geq\; m_\ell - 2\delta'
  \;>\; 0.
\]
Therefore $p^+_\ell \in \topB(\tilde{\pi}_\ell)$ over the full prefix
set~$\mathcal{P}_\ell$.

\textit{Subset argument.}\;
Actual beam search at level~$\ell$ considers only children of the~$B$
surviving prefixes from level~$\ell{-}1$, i.e., a subset
$\mathcal{P}'_\ell \subseteq \mathcal{P}_\ell$.
Since $S_{\ell-1}$ holds by the inductive hypothesis,
$p^+_{\ell-1}$ is in the beam, so $p^+_\ell$ is indeed a candidate
in~$\mathcal{P}'_\ell$.
If $p^+_\ell$ is in the top-$B$ of the full set~$\mathcal{P}_\ell$,
it is \emph{a fortiori} in the top-$B$ of any subset containing it,
because removing competitors can only improve its rank.

Applying the per-level argument at $\ell = 1, 2, \ldots, L$
completes the induction, establishing $S_L$.
\end{proof}

\subsection{Score fusion and decoder mismatch}
\label{app:score_fusion}

We formalize when geometric score fusion can reduce the decoder
mismatch~$\varepsilon_\ell$.
Fix a query~$q$, a level~$\ell$, and a surviving parent prefix
$p=z^{1:\ell-1}$.
Let $\hat{x}_p$ denote the partial reconstruction represented by this
parent prefix, and define the decoding residual
$r_p := q-\hat{x}_p$.
This residual is the query-to-prefix residual used at decoding time, not
the RQ assignment residual in Sec.~\ref{sec:preliminaries}.
Let $\mathcal{A}_\ell(p)$ denote the set of legal child tokens.
Each legal child token~$v \in \mathcal{A}_\ell(p)$ corresponds to a
codeword~$c_{\ell,v}$ and produces the candidate partial reconstruction
$\hat{x}_p+c_{\ell,v}$.
All conditional distributions below are PMFs over the finite token set
$\mathcal{A}_\ell(p)$; lower-case symbols such as $u$ and $v$ denote
token values.  When an expectation or covariance is used, the associated
upper-case symbol denotes a random token drawn from the stated PMF.
The conditional quantized oracle assigns higher probability to legal
tokens whose candidate reconstructions have larger similarity to the
query, while the conditional decoder distribution is produced by the
autoregressive decoder over the same legal tokens.
Decoder mismatch measures the distance between these two conditional
distributions.
Score fusion adds a codebook-derived bias to the decoder log-probability
so that candidate tokens reducing the query residual receive higher beam
scores.

The conditional quantized oracle over $\mathcal{A}_\ell(p)$ is
\[
  \bar{\pi}_\ell(v \mid p, q)
  =
  \frac{\exp\!\bigl(q^\top c_{\ell,v}/\tau\bigr)}
       {\sum_{u \in \mathcal{A}_\ell(p)}
        \exp\!\bigl(q^\top c_{\ell,u}/\tau\bigr)}.
\]
The parent term $q^\top \hat{x}_p$ is omitted because it is constant
across all children of the same parent prefix.
The conditional decoder distribution is
$\tilde{\pi}_\ell(v \mid p, q)$, and the geometric score used by score
fusion is
\[
  g_\ell(v)
  =
  2 r_p^\top c_{\ell,v}
  -
  c_{\ell,v}^\top c_{\ell,v}.
\]
Equivalently,
\[
  g_\ell(v)
  =
  \|r_p\|^2-\|r_p-c_{\ell,v}\|^2,
\]
so $g_\ell(v)$ is the reduction in squared distance from the query to
the partial reconstruction after appending token~$v$.

We now rewrite $g_\ell(v)$ in terms of the conditional quantized oracle.
Let
\[
  Z_\ell(p,q)
  =
  \sum_{u \in \mathcal{A}_\ell(p)}
  \exp\!\bigl(q^\top c_{\ell,u}/\tau\bigr).
\]
Then
\[
  q^\top c_{\ell,v}
  =
  \tau \log \bar{\pi}_\ell(v \mid p, q)
  +
  \tau \log Z_\ell(p,q).
\]
Substituting $r_p=q-\hat{x}_p$ into $g_\ell(v)$ gives
\begin{align}
  g_\ell(v)
  &=
  2(q-\hat{x}_p)^\top c_{\ell,v}
  -
  c_{\ell,v}^\top c_{\ell,v}
  \notag \\
  &=
  2q^\top c_{\ell,v}
  -
  \bigl(2\hat{x}_p^\top c_{\ell,v}
  +
  c_{\ell,v}^\top c_{\ell,v}\bigr)
  \notag \\
  &=
  2\tau \log \bar{\pi}_\ell(v \mid p, q)
  +
  \kappa_\ell(p,q)
  -
  \eta_\ell(v),
\label{eq:g_decomp}
\end{align}
where
\[
  \kappa_\ell(p,q) := 2\tau\log Z_\ell(p,q),
  \qquad
  \eta_\ell(v) := 2\hat{x}_p^\top c_{\ell,v}
  +
  c_{\ell,v}^\top c_{\ell,v}.
\]
The term $\kappa_\ell(p,q)$ is constant over child tokens of the same
parent prefix, whereas $\eta_\ell(v)$ is token-dependent.
Consequently, the fused distribution takes the form
\begin{equation}
  \pi^\omega_\ell(v \mid p, q)
  \;\propto\;
  \tilde{\pi}_\ell(v \mid p, q)\;\cdot\;
  \bar{\pi}_\ell(v \mid p, q)^{2\tau\omega}\;\cdot\;
  e^{-\omega\,\eta_\ell(v)}.
\label{eq:poe}
\end{equation}

\begin{proposition}[Geometric expert approximates quantized oracle]
\label{prop:geom_expert}
Define the geometric expert distribution
\[
  \rho_\ell(v \mid p, q)
  \;:=\;
  \frac{\exp\!\bigl(g_\ell(v) / (2\tau)\bigr)}
       {\sum_{u \in \mathcal{A}_\ell(p)}
        \exp\!\bigl(g_\ell(u) / (2\tau)\bigr)}.
\]
Then
\[
  \rho_\ell(v \mid p, q)
  \;=\;
  \frac{\bar{\pi}_\ell(v \mid p, q)\,
        \exp\!\bigl(-\eta_\ell(v)/(2\tau)\bigr)}
       {\sum_{u \in \mathcal{A}_\ell(p)}
        \bar{\pi}_\ell(u \mid p, q)\,
        \exp\!\bigl(-\eta_\ell(u)/(2\tau)\bigr)},
\]
where $\eta_\ell(v) := 2\hat{x}_p^\top c_{\ell,v}
+ c_{\ell,v}^\top c_{\ell,v}$.
If the distortion oscillation satisfies
$\mathrm{osc}(\eta_\ell) :=
\max_v \eta_\ell(v) - \min_v \eta_\ell(v) \leq \delta_\ell$,
where this oscillation measures the token-dependent deviation between
the geometric score and the conditional-oracle log-probability after
removing the parent-wise constant,
then
\begin{equation}
  \KL\!\bigl(\bar{\pi}_\ell(\cdot \mid p,q) \,\big\|\,
             \rho_\ell(\cdot \mid p,q)\bigr)
  \;\leq\; \frac{\delta_\ell}{2\tau},
  \qquad
  \TV\!\bigl(\bar{\pi}_\ell(\cdot \mid p,q),\,
             \rho_\ell(\cdot \mid p,q)\bigr)
  \;\leq\; \sqrt{\frac{\delta_\ell}{4\tau}}.
\label{eq:geom_expert_bound}
\end{equation}
\end{proposition}

\begin{proof}
The identity for~$\rho_\ell$ follows directly from
$g_\ell(v) = 2\tau \log \bar{\pi}_\ell(v \mid p,q)
+ \kappa_\ell(p,q) - \eta_\ell(v)$
(Eq.~\ref{eq:g_decomp}): dividing by~$2\tau$ and
exponentiating cancels $\kappa_\ell/(2\tau)$ in the normalization.

For the KL bound, write
$\bar{\eta} :=
\EE_{V \sim \bar{\pi}_\ell(\cdot \mid p,q)}[\eta_\ell(V)]$ and
$\Delta_v := \eta_\ell(v) - \bar{\eta}$ for each
$v \in \mathcal{A}_\ell(p)$.
Then
\[
  \KL\!\bigl(\bar{\pi}_\ell(\cdot \mid p,q) \,\big\|\,
             \rho_\ell(\cdot \mid p,q)\bigr)
  \;=\; \log
        \EE_{V \sim \bar{\pi}_\ell(\cdot \mid p,q)}\!\Bigl[
          e^{-\Delta_V / (2\tau)}\Bigr].
\]
Since $\max_v(-\Delta_v) \leq \mathrm{osc}(\eta_\ell) \leq \delta_\ell$,
every term in the expectation is at most
$e^{\delta_\ell/(2\tau)}$, giving
$\KL(\bar{\pi}_\ell(\cdot \mid p,q) \| \rho_\ell(\cdot \mid p,q))
\leq \delta_\ell / (2\tau)$.
The TV bound follows from Pinsker's inequality
(Lem.~\ref{lem:app_pinsker}).
\end{proof}

Prop.~\ref{prop:geom_expert} explains why the geometric score is
aligned with the conditional quantized oracle when the token-dependent
distortion is small.
The next result states when adding this score to the decoder reduces the
local decoder mismatch.

\begin{proposition}[Local reduction of decoder mismatch]
\label{prop:local_tv}
Define the fused family
\[
  \pi^\omega_\ell(v \mid p, q)
  \;:=\;
  \frac{\tilde{\pi}_\ell(v \mid p, q)\,
        \exp\!\bigl(\omega\, g_\ell(v)\bigr)}
       {\sum_{u \in \mathcal{A}_\ell(p)}
        \tilde{\pi}_\ell(u \mid p, q)\,
        \exp\!\bigl(\omega\, g_\ell(u)\bigr)},
  \qquad \omega \geq 0.
\]
Assume there is no tie between the conditional oracle and decoder
probabilities, i.e.,
\[
  \bar{\pi}_\ell(v \mid p,q)
  \neq
  \tilde{\pi}_\ell(v \mid p,q)
  \qquad \text{for every } v \in \mathcal{A}_\ell(p),
\]
and let
\[
  S^+_\ell
  \;:=\; \bigl\{v \in \mathcal{A}_\ell(p) :\;
         \bar{\pi}_\ell(v \mid p,q) > \tilde{\pi}_\ell(v \mid p,q)
         \bigr\}.
\]
Let $V \sim \tilde{\pi}_\ell(\cdot \mid p,q)$ be a random legal token
drawn from the decoder distribution.
Then
\begin{equation}
  \left.\frac{d}{d\omega}\,
    \TV\!\bigl(\bar{\pi}_\ell(\cdot \mid p,q),\,
               \pi^\omega_\ell(\cdot \mid p,q)\bigr)
  \right|_{\omega=0}
  \;=\;
  -\,\mathrm{Cov}\!\Bigl(
      g_\ell(V),\; \mathbf{1}\{V \in S^+_\ell\}
  \Bigr).
\label{eq:tv_deriv}
\end{equation}
In particular, if
$\mathrm{Cov}\!\bigl(
   g_\ell(V),\, \mathbf{1}\{V \in S^+_\ell\}\bigr) > 0$,
then there exists $\omega_0 > 0$ such that for all
$0 < \omega \leq \omega_0$,
\[
  \TV\!\bigl(\bar{\pi}_\ell(\cdot \mid p,q),\,
             \pi^\omega_\ell(\cdot \mid p,q)\bigr)
  \;<\;
  \TV\!\bigl(\bar{\pi}_\ell(\cdot \mid p,q),\,
             \tilde{\pi}_\ell(\cdot \mid p,q)\bigr).
\]
\end{proposition}

\begin{proof}
For each legal token $v \in \mathcal{A}_\ell(p)$, define the probability
masses
\[
  \bar{p}_v := \bar{\pi}_\ell(v \mid p,q), \qquad
  \tilde{p}_v := \tilde{\pi}_\ell(v \mid p,q), \qquad
  p^\omega_v := \pi^\omega_\ell(v \mid p,q).
\]
We use $V_\omega \sim \pi^\omega_\ell(\cdot \mid p,q)$ for a random token
drawn from the fused distribution at parameter~$\omega$.

\textit{Step~1: derivative of $\pi^\omega$.}\;
Since $\pi^\omega_\ell(\cdot \mid p,q)$ is an exponential-family tilt
of~$\tilde{\pi}_\ell(\cdot \mid p,q)$,
\[
  \frac{d}{d\omega}\,p^\omega_v
  \;=\; p^\omega_v\,
        \bigl(g_\ell(v)
        - \EE_{V_\omega \sim \pi^\omega_\ell(\cdot \mid p,q)}
          [g_\ell(V_\omega)]\bigr).
\]
At $\omega = 0$, $p^0_v = \tilde{p}_v$ for every token~$v$, so
\[
  \left.\frac{d}{d\omega}\,p^\omega_v\right|_{\omega=0}
  =
  \tilde{p}_v
  \bigl(g_\ell(v)
  - \EE_{V \sim \tilde{\pi}_\ell(\cdot \mid p,q)}[g_\ell(V)]
  \bigr).
\]

\textit{Step~2: sign preservation.}\;
By the no-tie assumption and the finite support,
$\bar{p}_v - \tilde{p}_v \neq 0$ for every~$v$.
By continuity of $p^\omega_v$ in~$\omega$, there exists $\epsilon > 0$
such that for all $0 < \omega < \epsilon$,
\[
  \mathrm{sign}\!\bigl(\bar{p}_v - p^\omega_v\bigr)
  \;=\;
  \mathrm{sign}\!\bigl(\bar{p}_v - \tilde{p}_v\bigr)
  \qquad
  \text{for every } v.
\]

\textit{Step~3: TV derivative.}\;
In this regime,
\[
  \TV\!\bigl(\bar{\pi}_\ell(\cdot \mid p,q),
             \pi^\omega_\ell(\cdot \mid p,q)\bigr)
  =
  \sum_{v \in S^+_\ell}\!\bigl(\bar{p}_v - p^\omega_v\bigr).
\]
Differentiating at $\omega = 0$:
\begin{align*}
  &\left.\frac{d}{d\omega}\,
  \TV\!\bigl(\bar{\pi}_\ell(\cdot \mid p,q),
             \pi^\omega_\ell(\cdot \mid p,q)\bigr)
  \right|_{\omega=0}
  \\
  &=
  -\sum_{v \in S^+_\ell}
     \tilde{p}_v\,
     \bigl(g_\ell(v)
     - \EE_{V \sim \tilde{\pi}_\ell(\cdot \mid p,q)}[g_\ell(V)]
     \bigr)  \\
  &=
  -\,\mathrm{Cov}_{V \sim \tilde{\pi}_\ell(\cdot \mid p,q)}\!\bigl(
        g_\ell(V),\; \mathbf{1}\{V \in S^+_\ell\}\bigr).
\end{align*}
If this covariance is positive, the derivative is strictly negative,
so by continuity there exists $\omega_0 > 0$ with
\[
  \TV\!\bigl(\bar{\pi}_\ell(\cdot \mid p,q),
             \pi^\omega_\ell(\cdot \mid p,q)\bigr)
  <
  \TV\!\bigl(\bar{\pi}_\ell(\cdot \mid p,q),
             \tilde{\pi}_\ell(\cdot \mid p,q)\bigr)
\]
for all $0 < \omega \leq \omega_0$.
\end{proof}

\begin{corollary}[Oracle-alignment decomposition]
\label{cor:oracle_align}
Under the decomposition
$g_\ell(v) = 2\tau \log \bar{\pi}_\ell(v \mid p,q)
+ \kappa_\ell - \eta_\ell(v)$,
let $V \sim \tilde{\pi}_\ell(\cdot \mid p,q)$.  Then
\begin{align}
  \mathrm{Cov}\!\bigl(
    g_\ell(V), \mathbf{1}\{V \in S^+_\ell\}\bigr)
  &=
  \underbrace{2\tau\,\mathrm{Cov}\!\bigl(
    \log \bar{\pi}_\ell(V \mid p,q),\,
    \mathbf{1}\{V \in S^+_\ell\}\bigr)
  }_{\text{oracle alignment}}
  \notag \\
  &\quad -
  \underbrace{\mathrm{Cov}\!\bigl(
    \eta_\ell(V),\, \mathbf{1}\{V \in S^+_\ell\}\bigr)
  }_{\text{distortion interference}}.
\label{eq:cov_decomp}
\end{align}
\end{corollary}

This decomposition makes the sufficient condition in
Prop.~\ref{prop:local_tv} easier to interpret.
The first covariance is positive when tokens underweighted by the decoder
tend to have high conditional-oracle probability.
The second covariance is the token-dependent distortion introduced by
$\eta_\ell(v)$.
Therefore, score fusion locally reduces $\varepsilon_\ell$ when the
oracle-alignment term dominates the distortion-interference term.
When $\varepsilon_\ell$ is large (as observed empirically in
App.~\ref{app:cross_dataset}, mean $\TV > 0.95$), many
oracle-preferred tokens are underweighted by the decoder.
In that case, $S^+_\ell$ concentrates on tokens with high
$\log \bar{\pi}_\ell(v \mid p,q)$, which makes the oracle-alignment term
large and positive.

\begin{remark}\rm
Props.~\ref{prop:geom_expert} and~\ref{prop:local_tv} operate on
the per-parent conditional distribution.
The global decoder mismatch
$\varepsilon_\ell = \TV(\bar{\pi}_\ell, \tilde{\pi}_\ell)$
in Thm.~\ref{thm:survival} is defined over all length-$\ell$
prefixes.
The local result therefore does not by itself prove a global reduction
for arbitrary beams.
It does explain the mechanism used in our decoder: if the alignment
condition in Prop.~\ref{prop:local_tv} holds across the surviving
parents, then reducing mismatch inside these conditional distributions
also reduces the global decoder mismatch over length-$\ell$ prefixes.
The empirical decoder mismatch is consistently large across datasets
(App.~\ref{app:cross_dataset}: mean $\TV > 0.95$, std ${<}\,0.03$),
which supports treating this mismatch as a systematic decoding issue
rather than an isolated parent-prefix case.
\end{remark}

\section{Additional Experimental Results}
\label{app:additional_exp}

\subsection{Generality to generative recommendation}
\label{app:gen_rec}

Tab.~\ref{tab:gen_rec} reports the additional experiment on TIGER~\citep{rajput2023recommender}.
We evaluate on the Amazon Beauty, Toys, and Sports categories~\citep{hou2024bridging}.
PRO improves all Recall and NDCG metrics over TIGER, which supports the generality of our analysis of prefix survival and PRO in RQ-based generative recommendation.

\begin{table}[t]
\caption{Comparison with TIGER on RQ-based generative recommendation.}
\label{tab:gen_rec}
\vspace{1mm}
\centering
\small
\setlength{\tabcolsep}{3.2pt}
\resizebox{\textwidth}{!}{
\begin{tabular}{l cccc cccc cccc}
\toprule
 & \multicolumn{4}{c}{\textbf{Beauty}} & \multicolumn{4}{c}{\textbf{Toys}} & \multicolumn{4}{c}{\textbf{Sports}} \\
\cmidrule(lr){2-5} \cmidrule(lr){6-9} \cmidrule(lr){10-13}
\textbf{Method} & R@5 & R@10 & N@5 & N@10 & R@5 & R@10 & N@5 & N@10 & R@5 & R@10 & N@5 & N@10 \\
\midrule
TIGER & 0.039038 & 0.055270 & 0.026157 & 0.031408 & 0.035906 & 0.052596 & 0.024123 & 0.029479 & 0.021771 & 0.033092 & 0.014115 & 0.017762 \\
PRO & \textbf{0.042794} & \textbf{0.062961} & \textbf{0.028821} & \textbf{0.035295} & \textbf{0.040645} & \textbf{0.060066} & \textbf{0.026946} & \textbf{0.033198} & \textbf{0.024046} & \textbf{0.037474} & \textbf{0.015183} & \textbf{0.019501} \\
\midrule
$\Delta$ (abs.) & +0.003756 & +0.007691 & +0.002664 & +0.003887 & +0.004739 & +0.007470 & +0.002823 & +0.003719 & +0.002275 & +0.004382 & +0.001068 & +0.001739 \\
$\Delta$ (rel.) & +9.62\% & +13.92\% & +10.18\% & +12.38\% & +13.20\% & +14.20\% & +11.70\% & +12.62\% & +10.45\% & +13.24\% & +7.57\% & +9.79\% \\
\bottomrule
\end{tabular}
}
\end{table}

\FloatBarrier
\subsection{Full retrieval results}
\label{app:full_results}

Tabs.~\ref{tab:main_full}--\ref{tab:main_full_3} report the
complete retrieval results on M-BEIR, extending Tab.~\ref{tab:main}
with R@10, grouped by query type.


\begin{table}[h!]
\caption{Full retrieval results on M-BEIR: text$\leftrightarrow$image
tasks. \textbf{Bold}: best generative method per column.}
\label{tab:main_full}
\vspace{1mm}
\centering
\setlength{\tabcolsep}{3pt}
\resizebox{0.95\textwidth}{!}{%
\begin{tabular}{l ccc ccc ccc}
\toprule
\multirow{2}{*}{\textbf{Method}}
& \multicolumn{3}{c}{COCO $t{\to}i$}
& \multicolumn{3}{c}{Flickr30k $t{\to}i$}
& \multicolumn{3}{c}{COCO $i{\to}t$} \\
\cmidrule(lr){2-4} \cmidrule(lr){5-7} \cmidrule(lr){8-10}
& R@1 & R@5 & R@10
& R@1 & R@5 & R@10
& R@1 & R@5 & R@10 \\
\midrule
\multicolumn{10}{l}{\textit{Dense}} \\
CLIP-SF
  & \textbf{55.2} & \textbf{80.7} & \textbf{88.4}
  & \textbf{79.1} & \textbf{95.0} & \textbf{97.5}
  & 65.7 & 87.5 & 92.6 \\
BLIP-FF
  & 53.7 & 79.6 & 87.3
  & 73.9 & 92.4 & 95.3
  & \textbf{71.1} & \textbf{90.9} & \textbf{95.4} \\
\midrule
\multicolumn{10}{l}{\textit{Generative}} \\
IRGen
  & 29.6 & 50.7 & 56.3
  & 49.0 & 68.9 & 72.5
  & -- & -- & -- \\
GRACE
  & 16.7 & 39.5 & 50.3
  & 37.4 & 59.5 & 66.2
  & -- & -- & -- \\
AVG
  & 31.3 & 58.0 & --
  & 62.8 & 85.4 & --
  & -- & -- & -- \\
SemCORE
  & 42.4 & 57.5 & --
  & 69.0 & 83.0 & --
  & -- & -- & -- \\
GENIUS
  & 37.1{\scriptsize$\pm$.33} & 66.4{\scriptsize$\pm$.28} & 76.7{\scriptsize$\pm$.22}
  & 56.8{\scriptsize$\pm$.42} & 81.3{\scriptsize$\pm$.25} & 85.5{\scriptsize$\pm$.18}
  & 47.9{\scriptsize$\pm$.35} & 77.3{\scriptsize$\pm$.24} & 86.0{\scriptsize$\pm$.17} \\
Ours
  & \textbf{45.8}{\scriptsize$\pm$.38} & \textbf{72.6}{\scriptsize$\pm$.30} & \textbf{81.3}{\scriptsize$\pm$.22}
  & \textbf{70.6}{\scriptsize$\pm$.45} & \textbf{87.3}{\scriptsize$\pm$.22} & \textbf{89.2}{\scriptsize$\pm$.15}
  & \textbf{59.0}{\scriptsize$\pm$.40} & \textbf{84.3}{\scriptsize$\pm$.25} & \textbf{90.9}{\scriptsize$\pm$.14} \\
\bottomrule
\end{tabular}%
}
\end{table}

\vspace{-3mm}
\begin{table}[h!]
\caption{Full retrieval results on M-BEIR (cont.):
text-to-multimodal and image-to-image tasks.
``--'': method does not support the task.}
\label{tab:main_full_2}
\vspace{1mm}
\centering
\setlength{\tabcolsep}{3pt}
\resizebox{0.95\textwidth}{!}{%
\begin{tabular}{l ccc ccc ccc}
\toprule
\multirow{2}{*}{\textbf{Method}}
& \multicolumn{3}{c}{WebQA $t{\to}t$}
& \multicolumn{3}{c}{WebQA $t{\to}i\text{,}t$}
& \multicolumn{3}{c}{NIGHTS $i{\to}i$} \\
\cmidrule(lr){2-4} \cmidrule(lr){5-7} \cmidrule(lr){8-10}
& R@1 & R@5 & R@10
& R@1 & R@5 & R@10
& R@1 & R@5 & R@10 \\
\midrule
\multicolumn{10}{l}{\textit{Dense}} \\
CLIP-SF
  & \textbf{58.3} & \textbf{84.2} & \textbf{90.0}
  & 47.6 & 76.3 & 83.9
  & \textbf{8.8}  & \textbf{31.6} & 52.8 \\
BLIP-FF
  & 51.9 & 78.6 & 84.9
  & \textbf{49.5} & \textbf{78.1} & \textbf{86.7}
  & 8.1  & 31.0 & \textbf{54.8} \\
\midrule
\multicolumn{10}{l}{\textit{Generative}} \\
IRGen
  & -- & -- & --
  & -- & -- & --
  & -- & -- & -- \\
GRACE
  & -- & -- & --
  & -- & -- & --
  & -- & -- & -- \\
GENIUS
  & 19.4{\scriptsize$\pm$.32} & 29.9{\scriptsize$\pm$.28} & 32.2{\scriptsize$\pm$.25}
  & 29.1{\scriptsize$\pm$.36} & 51.3{\scriptsize$\pm$.28} & 58.2{\scriptsize$\pm$.24}
  & 1.2{\scriptsize$\pm$.14}  & 11.6{\scriptsize$\pm$.28} & 30.7{\scriptsize$\pm$.35} \\
Ours
  & \textbf{25.7}{\scriptsize$\pm$.40} & \textbf{40.7}{\scriptsize$\pm$.35} & \textbf{45.2}{\scriptsize$\pm$.30}
  & \textbf{35.3}{\scriptsize$\pm$.42} & \textbf{57.6}{\scriptsize$\pm$.32} & \textbf{64.6}{\scriptsize$\pm$.26}
  & \textbf{3.5}{\scriptsize$\pm$.20}  & \textbf{19.7}{\scriptsize$\pm$.35} & \textbf{40.7}{\scriptsize$\pm$.38} \\
\bottomrule
\end{tabular}%
}
\end{table}

\vspace{-3mm}
\begin{table}[h!]
\caption{Full retrieval results on M-BEIR (cont.):
composed image$+$text queries.
``--'': method does not support the task.}
\label{tab:main_full_3}
\vspace{1mm}
\centering
\setlength{\tabcolsep}{5.5pt}
\resizebox{0.95\textwidth}{!}{%
\begin{tabular}{l ccc ccc ccc}
\toprule
\multirow{2}{*}{\textbf{Method}}
& \multicolumn{3}{c}{CIRR $i\text{,}t{\to}i$}
& \multicolumn{3}{c}{OVEN $i\text{,}t{\to}i\text{,}t$}
& \multicolumn{3}{c}{InfoSeek $i\text{,}t{\to}t$} \\
\cmidrule(lr){2-4} \cmidrule(lr){5-7} \cmidrule(lr){8-10}
& R@1 & R@5 & R@10
& R@1 & R@5 & R@10
& R@1 & R@5 & R@10 \\
\midrule
\multicolumn{10}{l}{\textit{Dense}} \\
CLIP-SF
  & 5.9  & 43.2 & 56.5
  & \textbf{49.4} & \textbf{69.2} & \textbf{75.3}
  & \textbf{14.5} & \textbf{28.8} & \textbf{37.0} \\
BLIP-FF
  & \textbf{24.4} & \textbf{50.9} & \textbf{61.3}
  & 36.1 & 56.4 & 63.2
  & 10.5 & 23.4 & 30.2 \\
\midrule
\multicolumn{10}{l}{\textit{Generative}} \\
IRGen
  & -- & -- & --
  & -- & -- & --
  & -- & -- & -- \\
GRACE
  & -- & -- & --
  & -- & -- & --
  & -- & -- & -- \\
GENIUS
  & 5.5{\scriptsize$\pm$.24}  & 19.0{\scriptsize$\pm$.30} & 28.1{\scriptsize$\pm$.28}
  & 26.4{\scriptsize$\pm$.38} & 36.0{\scriptsize$\pm$.30} & 43.9{\scriptsize$\pm$.26}
  & 6.1{\scriptsize$\pm$.22}  & 10.2{\scriptsize$\pm$.25} & 12.9{\scriptsize$\pm$.22} \\
Ours
  & \textbf{8.6}{\scriptsize$\pm$.30}  & \textbf{27.8}{\scriptsize$\pm$.36} & \textbf{38.2}{\scriptsize$\pm$.32}
  & \textbf{34.5}{\scriptsize$\pm$.44} & \textbf{51.3}{\scriptsize$\pm$.35} & \textbf{61.3}{\scriptsize$\pm$.28}
  & \textbf{7.3}{\scriptsize$\pm$.25}  & \textbf{15.0}{\scriptsize$\pm$.30} & \textbf{19.3}{\scriptsize$\pm$.28} \\
\bottomrule
\end{tabular}%
}
\end{table}

\FloatBarrier
\subsection{Complete ablation}
\label{app:full_ablation}

Tab.~\ref{tab:ablation_full} reports all seven configurations tested
in the ablation study, complementing the additive summary in
Tab.~\ref{tab:ablation}.
We organize the results by individual components, pairwise
combinations, and the full system.

\begin{table}[h]
\caption{Complete ablation results. Individual components (top),
pairwise combinations (middle), and full approach (bottom).
Abbreviations: PD: Prefix distillation; VS: Vocabulary scheduling; SF: Score fusion.
\textbf{Bold}: best per column. All entries: mean$\pm$std over 5 seeds.}
\label{tab:ablation_full}
\vspace{1mm}
\centering
\resizebox{\textwidth}{!}{%
\footnotesize
\setlength{\tabcolsep}{3pt}
\begin{tabular}{l cc cc cc cc cc cc}
\toprule
\multirow{3}{*}{\textbf{Configuration}}
& \multicolumn{2}{c}{COCO}
& \multicolumn{2}{c}{COCO}
& \multicolumn{2}{c}{WebQA}
& \multicolumn{2}{c}{WebQA}
& \multicolumn{2}{c}{OVEN}
& \multicolumn{2}{c}{OVEN} 
\\
& \multicolumn{2}{c}{$t {\to} i$}
& \multicolumn{2}{c}{$i {\to} t$}
& \multicolumn{2}{c}{$t {\to} t$}
& \multicolumn{2}{c}{$t {\to} i\text{,}t$}
& \multicolumn{2}{c}{$i\text{,}t {\to} i$}
& \multicolumn{2}{c}{$i\text{,}t {\to} i\text{,}t$} \\
\cmidrule(lr){2-3} \cmidrule(lr){4-5} \cmidrule(lr){6-7}
\cmidrule(lr){8-9} \cmidrule(lr){10-11} \cmidrule(lr){12-13}
& R@1 & R@5
& R@1 & R@5
& R@1 & R@5
& R@1 & R@5
& R@1 & R@5
& R@1 & R@5 \\
\midrule
\multicolumn{13}{l}{\textit{Baseline}} \\
Base (vanilla RQ)
  & 36.8{\scriptsize$\pm$.30} & 66.5{\scriptsize$\pm$.24}
  & 48.1{\scriptsize$\pm$.32} & 77.2{\scriptsize$\pm$.22}
  & 19.5{\scriptsize$\pm$.34} & 30.0{\scriptsize$\pm$.28}
  & 28.9{\scriptsize$\pm$.35} & 51.4{\scriptsize$\pm$.26}
  & 5.9{\scriptsize$\pm$.24} & 14.7{\scriptsize$\pm$.27}
  & 26.7{\scriptsize$\pm$.32} & 36.4{\scriptsize$\pm$.26} \\
\midrule
\multicolumn{13}{l}{\textit{Individual components}} \\
+ PD ($K_\ell{\downarrow}$)
  & 40.2{\scriptsize$\pm$.33} & 68.2{\scriptsize$\pm$.26}
  & 51.5{\scriptsize$\pm$.35} & 78.9{\scriptsize$\pm$.23}
  & 21.2{\scriptsize$\pm$.36} & 31.7{\scriptsize$\pm$.30}
  & 31.4{\scriptsize$\pm$.38} & 53.1{\scriptsize$\pm$.28}
  & 5.6{\scriptsize$\pm$.23} & 14.2{\scriptsize$\pm$.26}
  & 28.4{\scriptsize$\pm$.34} & 39.9{\scriptsize$\pm$.27} \\
+ VS\ ($m_\ell{\uparrow}$)
  & 38.5{\scriptsize$\pm$.32} & 68.6{\scriptsize$\pm$.25}
  & 49.4{\scriptsize$\pm$.33} & 79.5{\scriptsize$\pm$.22}
  & 19.1{\scriptsize$\pm$.35} & 29.6{\scriptsize$\pm$.29}
  & 30.1{\scriptsize$\pm$.36} & 52.7{\scriptsize$\pm$.27}
  & 8.5{\scriptsize$\pm$.28} & 18.5{\scriptsize$\pm$.30}
  & 30.4{\scriptsize$\pm$.35} & 44.1{\scriptsize$\pm$.28} \\
+ SF ($\varepsilon_\ell{\downarrow}$)
  & 39.7{\scriptsize$\pm$.34} & 69.8{\scriptsize$\pm$.26}
  & 51.7{\scriptsize$\pm$.36} & 81.1{\scriptsize$\pm$.24}
  & 23.4{\scriptsize$\pm$.40} & 36.6{\scriptsize$\pm$.33}
  & 31.2{\scriptsize$\pm$.38} & 54.3{\scriptsize$\pm$.29}
  & 5.3{\scriptsize$\pm$.22} & 13.3{\scriptsize$\pm$.25}
  & 28.6{\scriptsize$\pm$.33} & 40.7{\scriptsize$\pm$.28} \\
\midrule
\multicolumn{13}{l}{\textit{Pairwise combinations}} \\
+ PD + SF
  & 42.9{\scriptsize$\pm$.36} & 71.4{\scriptsize$\pm$.27}
  & 55.3{\scriptsize$\pm$.38} & 82.5{\scriptsize$\pm$.24}
  & 24.9{\scriptsize$\pm$.42} & \textbf{40.1}{\scriptsize$\pm$.34}
  & 33.7{\scriptsize$\pm$.40} & 55.9{\scriptsize$\pm$.30}
  & 6.2{\scriptsize$\pm$.25} & 14.5{\scriptsize$\pm$.27}
  & 30.2{\scriptsize$\pm$.35} & 44.8{\scriptsize$\pm$.29} \\
+ PD + VS
  & 42.1{\scriptsize$\pm$.35} & 70.7{\scriptsize$\pm$.26}
  & 53.6{\scriptsize$\pm$.37} & 81.4{\scriptsize$\pm$.23}
  & 21.6{\scriptsize$\pm$.38} & 33.2{\scriptsize$\pm$.31}
  & 33.5{\scriptsize$\pm$.39} & 55.2{\scriptsize$\pm$.29}
  & 8.8{\scriptsize$\pm$.28} & \textbf{19.0}{\scriptsize$\pm$.30}
  & 32.2{\scriptsize$\pm$.36} & 47.5{\scriptsize$\pm$.29} \\
\midrule
\multicolumn{13}{l}{\textit{Full approach}} \\
+ PD + VS + SF
  & \textbf{46.2}{\scriptsize$\pm$.38} & \textbf{73.0}{\scriptsize$\pm$.28}
  & \textbf{59.3}{\scriptsize$\pm$.42} & \textbf{84.6}{\scriptsize$\pm$.25}
  & \textbf{26.2}{\scriptsize$\pm$.44} & 39.7{\scriptsize$\pm$.34}
  & \textbf{35.6}{\scriptsize$\pm$.42} & \textbf{57.5}{\scriptsize$\pm$.32}
  & \textbf{9.1}{\scriptsize$\pm$.30} & 17.6{\scriptsize$\pm$.30}
  & \textbf{34.4}{\scriptsize$\pm$.38} & \textbf{51.4}{\scriptsize$\pm$.31} \\
\bottomrule
\end{tabular}%
}
\end{table}

\paragraph{Per-component analysis.}
Among individual components, the strongest contributor varies by task:
vocabulary scheduling gives the largest standalone gains on the OVEN
settings, score fusion is strongest on WebQA \(t\to t\), and prefix
distillation is strongest on COCO. These patterns are consistent with
the view from the survival bound that different tasks are dominated by
different quantities.

\paragraph{Pairwise interactions.}
Prefix distillation and vocabulary scheduling show robust positive
interaction across the ablation settings. This supports the hypothesis
that reducing ranking divergence makes subsequent margin enlargement
more effective.
In contrast, prefix distillation + score fusion shows mixed
interaction (positive on 3 tasks, near-zero or slightly negative on
3), confirming that the synergy between the two codebook-learning
components is more robust than the codebook--decoder interaction.
Notably, vocabulary scheduling alone hurts WebQA~$t{\to}t$ by
$-$0.4, but after prefix distillation the marginal vocabulary
scheduling gain turns positive ($+$0.4), consistent with the
hypothesis that reducing $K_\ell$ stabilizes subsequent margin
enlargement.

\paragraph{R@5 exceptions.}
The full system achieves the best R@5 on 4 of 6 tasks.
Two exceptions: prefix distillation + score fusion outperforms Full on
WebQA~$t{\to}t$ R@5 by $+$0.4, and prefix distillation +
vocabulary scheduling outperforms Full on OVEN~$i\text{,}t{\to}i$ R@5 by
$+$1.4.
Both exceptions involve tasks where score fusion or vocabulary
scheduling provide weaker marginal gains at R@5, suggesting that
inference-time score fusion is less effective on small text-only pools
and on composed queries where the geometric proxy is noisy.

\subsection{Cross-dataset diagnostic analysis}
\label{app:cross_dataset}

Sec.~\ref{sec:diagnostic} validates the survival bound's
predictions on COCO.
Fig.~\ref{fig:cross_dataset} extends this analysis to four
additional datasets spanning all three query families, using
$\tau{=}1.0$ and beam width~20.

The top row shows $K_\ell$ quintile versus survival at levels 0, 3,
and~5. All five datasets exhibit a consistent monotone decline: higher
ranking divergence predicts lower prefix survival at every level,
confirming that the relationship identified on COCO generalizes
across query modalities and dataset characteristics.
The effect is strongest on composed-query tasks (CIRR, OVEN), where
$K_\ell$ variance is largest, and mildest on Flickr30k, where the
overall survival rate is high and leaves less room for differentiation.

The bottom row shows $m_\ell$ quintile versus survival. All datasets
display a monotone increase: larger teacher margin correlates with
higher survival at every level. The effect is milder than for
$K_\ell$, consistent with the main-text finding that ranking
divergence is the dominant bottleneck in the survival bound.

\paragraph{Decoder mismatch is a systemic bottleneck.}
Unlike $K_\ell$ and $m_\ell$, which exhibit substantial per-query
variance, the decoder mismatch $\varepsilon_\ell$ is uniformly high
across all queries and all five datasets (mean
$\TV > 0.95$, std $< 0.03$). Consequently, quintile-binning
$\varepsilon_\ell$ reveals no monotone trend with survival---not
because decoder mismatch is unimportant, but because it affects
\emph{every} query equally.
This finding reframes the role of score fusion
(Sec.~\ref{sec:cdsb}): rather than selectively rescuing queries
with unusually high decoder mismatch, it compensates for a systemic
gap between the decoder scoring distribution and the quantized oracle, explaining why
score fusion yields broad-based improvements across all tasks and
query types (Tab.~\ref{tab:ablation}).

\begin{figure}[h]
  \centering
  \includegraphics[width=\textwidth]{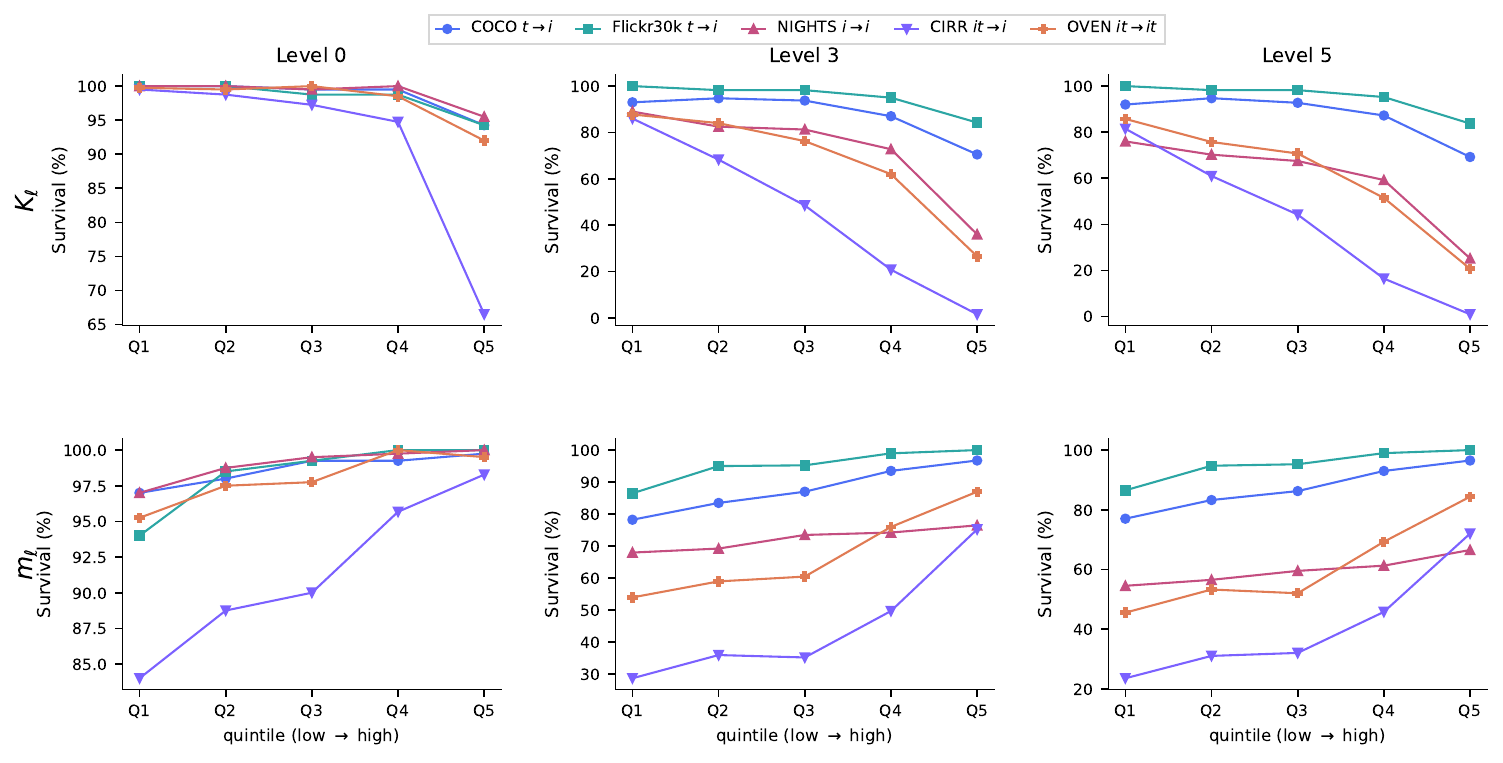}
  \caption{Cross-dataset diagnostic analysis.
    Top row: $K_\ell$ quintile vs.\ survival.
    Bottom row: $m_\ell$ quintile vs.\ survival.
    All five datasets show monotone trends consistent with the
    survival bound (Thm.~\ref{thm:survival}).}
  \label{fig:cross_dataset}
\end{figure}

\subsection{Efficiency comparison}
\label{app:efficiency}

\paragraph{Prefix survival and throughput scaling.}
Fig.~\ref{fig:efficiency}\subref{fig:oracle_prefix_survival} reports
the per-level survival rate of the target identifier token under quantized oracle scoring.
At each \ac{RQ} level, this metric conditions on the preceding target prefix and checks whether the target token at that level is ranked among the top-$B$ candidate tokens.
Fig.~\ref{fig:efficiency}\subref{fig:throughput} compares query
throughput as the candidate pool grows from 5K to 300K on a single GPU.
Dense retrieval stores all $N$ embeddings and scans them at query
time, with cost $O(N {\cdot} d)$; generative decoding stores only
discrete RQ codes and traverses a trie at $O(L {\cdot} B {\cdot} V)$,
independent of~$N$.
Dense throughput declines steadily as the pool grows, while generative
throughput remains constant.
The two curves cross near $N{=}50\text{K}$; beyond this point
generative retrieval is faster.
Adding score fusion (Eq.~\ref{eq:cdsb}) has virtually no effect on
throughput, as it only computes a lightweight residual inner product
per decoding step.

\begin{figure}[h]
  \centering
  \begin{subfigure}[t]{0.47\textwidth}
    \centering
    \includegraphics[width=\linewidth]{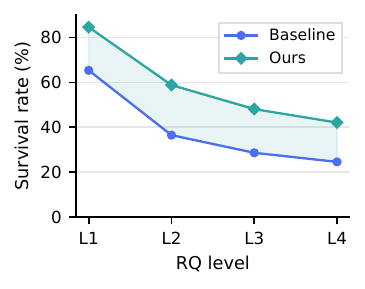}
    \caption{Per-level token survival}
    \label{fig:oracle_prefix_survival}
  \end{subfigure}\hfill
  \begin{subfigure}[t]{0.49\textwidth}
    \centering
    \includegraphics[width=\linewidth]{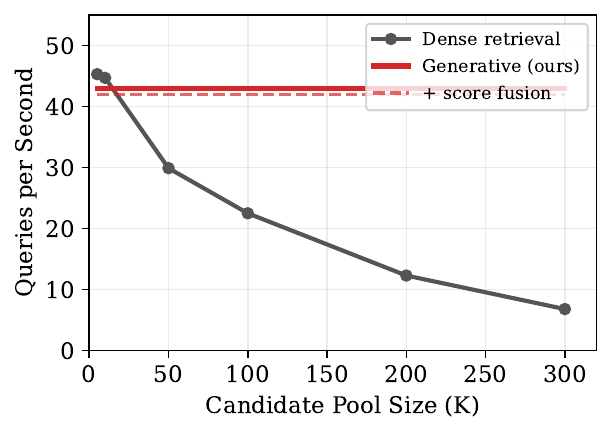}
    \caption{Throughput scaling}
    \label{fig:throughput}
  \end{subfigure}
  \caption{End-to-end diagnostics and efficiency.
    (a) Per-level survival rate of the target identifier token on COCO text-to-image under quantized oracle scoring; our method improves target token retention over the baseline indexing.
    (b) Query throughput vs.\ candidate pool size on a single GPU.
    Dense retrieval scales as $O(N {\cdot} d)$; generative decoding remains constant.
    Score fusion adds negligible overhead (the two generative curves nearly overlap).}
  \label{fig:efficiency}
\end{figure}

\paragraph{Index size.}
Tab.~\ref{tab:index_size} compares the index footprint at several
pool sizes. The generative index stores $L{=}16$ integer codes per
target item (two bytes each), yielding a constant compression ratio of
over ninety times relative to storing full 768-dimensional
floating-point embeddings.

\begin{table}[h]
\caption{Index size comparison. Dense retrieval stores float32
embeddings ($N {\times} 768 {\times} 4$\,B); generative retrieval
stores int16 RQ codes ($N {\times} 16 {\times} 2$\,B).}
\label{tab:index_size}
\vspace{1mm}
\centering
\begin{tabular}{r r r r}
\toprule
Pool size $N$ & Dense (MB) & Generative (MB) & Ratio \\
\midrule
5\,K    &   14.6 &  0.15 & 96$\times$ \\
50\,K   &  146.5 &  1.53 & 96$\times$ \\
335\,K  &  979.0 & 10.2  & 96$\times$ \\
1\,M    & 2\,930 & 30.5  & 96$\times$ \\
\bottomrule
\end{tabular}
\end{table}

\subsection{Encoder sensitivity}
\label{app:encoder_ablation}

To verify that the performance gap between RQ-based \ac{MGR} and dense
retrieval is not primarily driven by encoder choice, we compare three
multimodal encoders under vanilla GENIUS (RQ $L=16$, $V=4096$) with a
T5-small decoder (beam width 20, no re-ranking), across four M-BEIR
tasks (MSCOCO $t{\to}i$, Flickr30k $t{\to}i$, WebQA $t{\to}t$,
NIGHTS $i{\to}i$). The encoders span different capacities and training
regimes:
\begin{itemize}[leftmargin=1.8em]
  \item \textbf{CLIP-OpenAI} (ViT-L/14, 768-d): zero-shot CLIP without
  M-BEIR fine-tuning, serving as the non-instruct baseline.
  \item \textbf{LLM2CLIP}~\citep{huang2026llm2clip} (LLM2Vec-8B +
  ViT-L/14-336, 1280-d): higher-capacity encoder via an LLM backbone,
  also evaluated in the non-instruct regime.
  \item \textbf{CLIP-SF}~\citep{wei2023uniir} (ViT-L/14, 768-d):
  UniIR's instruction-tuned CLIP, used in its native instruct regime.
\end{itemize}

Tab.~\ref{tab:encoder_ablation} reports dense-retrieval (query-target
similarity), vanilla GENIUS, and our method Recall@\{1,5\}
across the four tasks.

\begin{table}[h]
\caption{Recall@\{1,5\} with three multimodal encoders on four M-BEIR
tasks. Dense: query-target inner-product similarity; GENIUS: vanilla GENIUS pipeline with the corresponding encoder; Ours: our pipeline with the corresponding encoder.}
\label{tab:encoder_ablation}
\vspace{1mm}
\centering
\setlength{\tabcolsep}{3pt}
\begin{tabular}{l l cc cc cc}
\toprule
\textbf{Task} & \textbf{Encoder} & \multicolumn{2}{c}{\textbf{Dense}} & \multicolumn{2}{c}{\textbf{GENIUS}} & \multicolumn{2}{c}{\textbf{Ours}} \\
\cmidrule(lr){3-4} \cmidrule(lr){5-6} \cmidrule(lr){7-8}
 &  & R@1 & R@5 & R@1 & R@5 & R@1 & R@5 \\
\midrule
\multirow{3}{*}{MSCOCO $t{\to}i$}
  & CLIP-OpenAI & 38.2 & 62.5 & 27.8 & 53.7 & 38.6 & 63.7 \\
  & LLM2CLIP    & 54.4 & 78.5 & 34.0 & 58.6 & 45.5 & 68.5 \\
  & CLIP-SF     & 55.2 & 80.7 & 37.1 & 66.4 & 45.8 & 72.6 \\
\midrule
\multirow{3}{*}{Flickr30k $t{\to}i$}
  & CLIP-OpenAI & 66.9 & 88.9 & 55.4 & 80.7 & 68.3 & 87.8 \\
  & LLM2CLIP    & 82.5 & 95.8 & 57.5 & 80.2 & 75.6 & 89.9 \\
  & CLIP-SF     & 79.1 & 95.0 & 56.8 & 81.3 & 70.6 & 87.3 \\
\midrule
\multirow{3}{*}{WebQA $t{\to}t$}
  & CLIP-OpenAI & 22.2 & 40.9 &  9.9 & 21.6 & 17.4 & 25.1 \\
  & LLM2CLIP    & 50.6 & 78.8 & 17.6 & 27.6 & 29.5 & 44.2 \\
  & CLIP-SF     & 58.3 & 84.2 & 19.4 & 29.9 & 25.7 & 40.7 \\
\midrule
\multirow{3}{*}{NIGHTS $i{\to}i$}
  & CLIP-OpenAI &  6.8 & 27.8 &  0.5 &  6.9 & 3.3 & 19.3 \\
  & LLM2CLIP    &  8.4 & 31.2 &  2.2 & 14.1 & 3.4 & 20.0 \\
  & CLIP-SF     &  8.8 & 31.6 &  1.2 & 11.6 & 3.5 & 19.7 \\
\bottomrule
\end{tabular}
\end{table}

Across encoders of different capacities and training regimes, the
vanilla GENIUS column exhibits a similar proportional gap between
generative and dense retrieval: switching to a higher-capacity encoder
or a retrieval-specialized one shifts both scores upward but does not
close their relative gap. This suggests that the gap is not caused by
insufficient encoder expressiveness; rather, the RQ-based
tokenization-and-decoding bottleneck prevents the decoder from fully
preserving the discrimination in the encoder space. The Ours column
consistently narrows this decay across encoders, indicating that our
method mitigates the RQ bottleneck rather than simply benefiting from a
particular encoder choice.

\section{Implementation Details}
\label{app:impl}

All experiments use the CLIP-SF encoder (ViT-L/14, 768-d
embeddings)~\citep{wei2023uniir} and T5-small
decoder~\citep{raffel2020exploring}. CLIP weights are frozen
throughout; only the RQ codebooks and T5 decoder parameters are trained. All training uses a single node with 4
NVIDIA RTX~4090 GPUs (24GB each) and PyTorch DDP.
Unless otherwise specified, results on the nine multimodal retrieval tasks are averaged over 5 random seeds.

\paragraph{Datasets.}
Tab.~\ref{tab:datasets} summarizes the 9 dataset--task pairs from
the M-BEIR benchmark~\citep{wei2023uniir} used in our evaluation.
Each task defines a query modality, a target modality, and a
task-specific candidate pool.
Pool sizes range from 1K (Flickr30k) to 612K (InfoSeek),
spanning three orders of magnitude.

\begin{table}[h]
\caption{Dataset statistics. All splits follow the M-BEIR
benchmark~\citep{wei2023uniir}. Pool sizes refer to the
task-specific test candidate pool.}
\label{tab:datasets}
\vspace{1mm}
\centering
\begin{tabular}{l l l l r r r}
\toprule
\textbf{Dataset} & \textbf{Task}
  & \textbf{Query} & \textbf{Target}
  & \textbf{Train} & \textbf{Test}
  & \textbf{Pool} \\
\midrule
COCO     & $t{\to}i$      & text   & image       & 100K  & 24.8K & 5.0K \\
Flickr30k& $t{\to}i$      & text   & image       & 145K  & 5.0K  & 1.0K \\
COCO     & $i{\to}t$      & image  & text        & 113K  & 5.0K  & 24.8K \\
NIGHTS   & $i{\to}i$      & image  & image       & 16K   & 2.1K  & 40K \\
\midrule
WebQA    & $t{\to}t$      & text   & text        & 16K   & 2.5K  & 544K \\
WebQA    & $t{\to}it$     & text   & image+text  & 17K   & 2.5K  & 403K \\
\midrule
CIRR     & $it{\to}i$     & image+text & image   & 26K   & 4.2K  & 22K \\
OVEN     & $it{\to}it$    & image+text & image+text & 305K & 14.7K & 335K \\
InfoSeek & $it{\to}t$     & image+text & text    & 284K  & 11.3K & 612K \\
\bottomrule
\end{tabular}
\end{table}

\subsection{Existing assets and licenses}
\label{app:asset_licenses}

Tab.~\ref{tab:data_asset_licenses} and
Tab.~\ref{tab:model_code_asset_licenses} summarize the existing
datasets, pretrained models, and code assets used in our experiments.
We use these assets only for research evaluation, cite their original
papers or release pages, and do not redistribute raw datasets or
pretrained checkpoints in the released code repository.

\begin{table}[h]
\caption{Dataset assets used in this paper.}
\label{tab:data_asset_licenses}
\vspace{1mm}
\centering
\scriptsize
\renewcommand{\arraystretch}{1.18}
\begin{tabularx}{\textwidth}{>{\raggedright\arraybackslash}p{0.24\textwidth} >{\raggedright\arraybackslash}p{0.30\textwidth} >{\raggedright\arraybackslash}p{0.22\textwidth} >{\raggedright\arraybackslash}X}
\toprule
\textbf{Asset} & \textbf{Use} & \textbf{License / terms} & \textbf{Source} \\
\midrule
M-BEIR & Main multimodal retrieval benchmark
& MIT for the benchmark release; underlying media and data follow the original source terms specified by M-BEIR
& \citet{wei2023uniir}; \href{https://huggingface.co/datasets/TIGER-Lab/M-BEIR}{Hugging Face} \\
Flickr30k & Additional text-to-image retrieval benchmark
& Images are obtained from Flickr and must follow the Flickr Terms of Use; the dataset maintainers state that the data are provided only for non-commercial research and educational purposes
& \citet{young2014flickr30k}; \href{https://shannon.cs.illinois.edu/DenotationGraph/}{dataset page} \\
Amazon Reviews'23 & Beauty, Toys, and Sports categories for the TIGER generalization experiment
& No explicit dataset license assigned by the release; made available primarily for research purposes; raw data are not redistributed in our package
& \citet{hou2024bridging}; \href{https://amazon-reviews-2023.github.io/}{dataset page}; \href{https://huggingface.co/datasets/McAuley-Lab/Amazon-Reviews-2023/discussions/1}{license discussion} \\
\bottomrule
\end{tabularx}
\end{table}

\begin{table}[h]
\caption{Pretrained model and code assets used in this paper.}
\label{tab:model_code_asset_licenses}
\vspace{1mm}
\centering
\scriptsize
\renewcommand{\arraystretch}{1.18}
\begin{tabularx}{\textwidth}{>{\raggedright\arraybackslash}p{0.25\textwidth} >{\raggedright\arraybackslash}p{0.31\textwidth} >{\raggedright\arraybackslash}p{0.16\textwidth} >{\raggedright\arraybackslash}X}
\toprule
\textbf{Asset} & \textbf{Use} & \textbf{License / terms} & \textbf{Source} \\
\midrule
UniIR CLIP-SF and BLIP-FF checkpoints & Dense baselines and frozen CLIP-SF encoder for GENIUS and PRO & MIT
& \citet{wei2023uniir}; \href{https://huggingface.co/TIGER-Lab/UniIR}{Hugging Face} \\
UniIR codebase & Reference code for M-BEIR data loading and evaluation conventions & MIT
& \citet{wei2023uniir}; \href{https://github.com/TIGER-AI-Lab/UniIR}{GitHub} \\
OpenAI CLIP ViT-L/14 & Encoder-sensitivity experiment & MIT
& \citet{radford2021learning}; \href{https://github.com/openai/CLIP}{GitHub} \\
T5-small & Autoregressive decoder backbone & Apache-2.0
& \citet{raffel2020exploring}; \href{https://huggingface.co/google-t5/t5-small}{Hugging Face} \\
LLM2CLIP & Encoder-sensitivity experiment & Apache-2.0
& \citet{huang2026llm2clip}; \href{https://huggingface.co/microsoft/LLM2CLIP-Openai-L-14-224}{Hugging Face} \\
GENIUS official implementation & Reference implementation for the reproduced GENIUS baseline & MIT
& \citet{genius2025}; \href{https://github.com/sung-yeon-kim/GENIUS-CVPR25}{GitHub} \\
\bottomrule
\end{tabularx}
\end{table}

\paragraph{Item codes.}
Each target item is represented by a sequence of $L{+}1{=}17$ discrete
tokens: one modality token ($V_{\mathrm{mod}}{=}3$) followed by $L{=}16$ RQ code tokens
with an ascending vocabulary schedule:
\[
  V_\ell =
  \begin{cases}
    512  & \ell = 0,\ldots,3    \quad\text{(4 levels)},\\
    1024 & \ell = 4,\ldots,11   \quad\text{(8 levels)},\\
    2048 & \ell \geq 12         \quad\text{(4 levels)}.
  \end{cases}
\]
The baseline (GENIUS) uses a uniform vocabulary $V{=}4096$ at all 16
levels with the same modality token.

\paragraph{Codebook-learning objective.}
For an input embedding~$y$ quantized by the RQ codebooks, let
$r_0=y$, $z^\ell=\argmin_v\|r_{\ell-1}-c_{\ell,v}\|^2$, and
$r_\ell=r_{\ell-1}-c_{\ell,z^\ell}$.
Following the standard RQ objective used in GENIUS~\citep{genius2025},
the residual-quantization loss is
\begin{equation}
  \mathcal{L}_{\mathrm{rq}}(y)
  \;=\;
  \sum_{\ell=1}^{L}
  \left\|r_{\ell-1}
  - \mathrm{sg}\!\left(c_{\ell,z^\ell}\right)\right\|_2^2,
\label{eq:rq_loss}
\end{equation}
where $\mathrm{sg}(\cdot)$ denotes the stop-gradient operator.
In codebook learning, $\mathcal{L}_{\mathrm{rq}}$ is averaged over the
query and target embeddings in the batch.
The full objective for codebook learning is
\begin{equation}
  \mathcal{L}_{\mathrm{cb}}
  \;=\; \mathcal{L}_{\mathrm{cl}}
  \;+\; \beta_{\mathrm{rq}}\,\mathcal{L}_{\mathrm{rq}}
  \;+\; \gamma_{\mathrm{mse}}\,\mathcal{L}_{\mathrm{mse}}
  \;+\; \lambda\,\mathcal{L}_{\mathrm{rank}}\,,
\label{eq:cb_full}
\end{equation}
where $\mathcal{L}_{\mathrm{cl}}$ is a bidirectional contrastive
loss on pre-quantization embeddings,
$\mathcal{L}_{\mathrm{rq}}$ is the residual-quantization loss between
residuals and assigned codewords,
$\mathcal{L}_{\mathrm{mse}}=\|\hat{q}-\hat{x}_j\|_2^2$ is the
MSE between quantized query and target vectors,
and $\mathcal{L}_{\mathrm{rank}}$ is the prefix ranking distillation
loss (Sec.~\ref{sec:rpq}).
We set $\beta_{\mathrm{rq}}{=}100$, $\gamma_{\mathrm{mse}}{=}100$, $\lambda{=}100$.

\paragraph{Codebook-learning schedule.}
Training runs for 20 epochs with learning rate $10^{-4}$ (Adam),
batch size 512 per GPU (effective batch 2048 with 4 GPUs),
and no learning rate warmup.
Prefix ranking distillation uses
$\lambda_{\mathrm{rank}}{=}100$, shared temperature
$\tau{=}0.05$, top-$k{=}128$ in-batch candidates, with 10\% linear warmup.

\paragraph{Decoder training.}
The T5-small decoder (60M parameters) is trained for 30 epochs with
learning rate $10^{-4}$ (Adam) and batch size 512.
Input query embeddings are projected to a single soft token via a
linear layer, then fed to the T5 encoder.
The decoder autoregressively generates the 17-token code sequence
(modality + 16 RQ tokens), supervised with per-token cross-entropy
loss.

\paragraph{Inference.}
We use trie-constrained beam search and geometric score fusion (Eq.~\ref{eq:cdsb}) adds a bias
$\omega{=}10$ at each decoding step, computed from the query residual
and the current-level codebook.
No reranking is applied after beam search.

\paragraph{Hyperparameters.}
Tab.~\ref{tab:hparam} summarizes all key hyperparameters.

\begin{table}[h]
\caption{Hyperparameter summary. All results averaged over 5 seeds.}
\label{tab:hparam}
\vspace{1mm}
\centering
\begin{tabular}{llc}
\toprule
\textbf{Component} & \textbf{Hyperparameter} & \textbf{Value} \\
\midrule
\multirow{4}{*}{Codebook (RQ)}
  & RQ levels $L$ & 16 \\
  & Vocab schedule $V_\ell$ & $512 / 1024 / 2048$ \\
  & Learning rate & $10^{-4}$ \\
  & Training epochs & 20 \\
\midrule
\multirow{4}{*}{Prefix distillation}
  & $\lambda_{\mathrm{rank}}$ & 100 \\
  & $\tau$ (shared) & 0.05 \\
  & Top-$k$ candidates & 128 \\
  & Warmup fraction & 10\% \\
\midrule
\multirow{3}{*}{Decoder (T5)}
  & Learning rate & $10^{-4}$ \\
  & Training epochs & 30 \\
  & Batch size (per GPU) & 512 \\
\midrule
\multirow{2}{*}{Inference}
  & Beam width $B$ & 50 \\
  & Score fusion $\omega$ & 10 \\
\bottomrule
\end{tabular}
\end{table}

\paragraph{Training overhead.}
Prefix ranking distillation adds a KL term over in-batch candidates at
each of the $L$ levels; no additional forward pass is required.
On 4$\times$RTX~4090 GPUs, the codebook-learning stage (Stage~1) takes
${\sim}4$\,min/epoch with distillation versus ${\sim}3.8$\,min/epoch
without, an overhead of ${\sim}8$\%.
Since decoder training (Stage~2, ${\sim}23$\,min/epoch) is unchanged
and dominates wall-clock time, the end-to-end training overhead is
under $1$\%.
Score fusion adds only a single inner product per decoding step at
inference and has negligible latency impact (Fig.~\ref{fig:efficiency}).

\paragraph{Vocabulary schedule design.}
The ascending schedule
$V_1 \leq V_2 \leq \cdots \leq V_L$ follows directly from the survival
bound (Eq.~\ref{eq:survival}): early levels face irreversible pruning,
so fewer codewords enlarge the teacher margin~$m_\ell$
(Sec.~\ref{sec:daca}); later levels contribute less to survival risk
and benefit from larger vocabularies for reconstruction fidelity.
We adopt a three-tier split ($512/1024/2048$) that allocates the
smallest vocabulary to the first 4 levels---where survival decay is
steepest (Sec.~\ref{sec:evidence})---and the largest to the final 4
levels.
The 8 intermediate levels use a moderate size as a transition.
This schedule was selected based on the principle above and not
tuned exhaustively; Tab.~\ref{tab:daca} confirms that even coarser
ascending schedules consistently outperform uniform and descending
alternatives.

\paragraph{Baselines.}
For IRGen~\citep{zhang2024irgen} and GRACE~\citep{li2024grace}, we report the
results reproduced by GENIUS~\citep{genius2025}.
For AVG~\citep{cai2025avg} and SemCORE~\citep{wu2025semcore}, we report the
results from SemCORE~\citep{wu2025semcore}.
We have verified that the test splits and candidate pools used in
these works are consistent with those of the M-BEIR benchmark.

\subsection{Statistical variability}
\label{app:ci}

GENIUS (reproduced) and our method are each trained with 5
independent random seeds.
Tabs.~\ref{tab:main_full}--\ref{tab:main_full_3} and
Tab.~\ref{tab:ablation_full} report per-cell mean$\pm$std.
Across all tasks, the maximum standard deviation is 0.45 (R@1)
and 0.36 (R@5); the corresponding 95\% confidence interval
half-widths ($1.96\,\sigma/\!\sqrt{5}$) are below 0.40 in all
cases.
All pairwise comparisons between our full system and GENIUS are
statistically significant ($p < 0.01$, two-sample $t$-test),
including the smallest gap (InfoSeek R@1, $\Delta {=} 1.2$,
$p < 0.01$).
Dense baselines (CLIP-SF, BLIP-FF) use deterministic pretrained
weights.

\begin{table}[h]
\caption{Hyperparameter sensitivity (COCO $t{\to}i$).
Default values in \textbf{bold}. Each panel varies one parameter with
the others fixed.}
\label{tab:sens}
\vspace{1mm}
\centering
\small
\begin{tabular}{cc cc cc}
\toprule
$\lambda_{\mathrm{rank}}$ & R@1 / R@5
  & $\tau$ & R@1 / R@5
  & $\omega$ & R@1 / R@5 \\
\midrule
0   & 41.3 / 70.2  & 0.01 & 40.3 / 68.7  & 0  & 42.1 / 70.7 \\
1   & 41.8 / 70.6  & 0.02 & 44.6 / 71.9  & 1  & 42.4 / 70.9 \\
10  & 44.7 / 72.4  & \textbf{0.05} & \textbf{46.2 / 73.0}
    & 3  & 43.9 / 72.3 \\
50  & 46.5 / 72.1  & 0.1  & 45.8 / 72.8  & 5  & 45.6 / 72.1 \\
\textbf{100} & \textbf{46.2 / 73.0}
    & 0.2  & 43.1 / 71.0  & \textbf{10} & \textbf{46.2 / 73.0} \\
200 & 44.8 / 72.5  & 0.5  & 39.7 / 67.4  & 20 & 45.5 / 73.4 \\
500 & 41.5 / 69.3  & 1.0  & 38.2 / 66.1  & 50 & 42.0 / 70.1 \\
\bottomrule
\end{tabular}
\end{table}

\subsection{Hyperparameter sensitivity}
\label{app:hparam_sens}

Tab.~\ref{tab:sens} reports Recall@\{1,\,5\} on COCO $t{\to}i$
(single seed) as each key hyperparameter is varied while the others are
held at their default values
($\lambda_{\mathrm{rank}}{=}100$, $\tau{=}0.05$, $\omega{=}10$).

All three parameters exhibit a broad plateau around the chosen
operating point.
For the distillation strength $\lambda_{\mathrm{rank}}$, values in
$[50, 200]$ yield R@1 within 1.7 points of the best; the dominant
effect is whether distillation is enabled at all ($\lambda{=}0$ vs.\
$\lambda{>}0$: ${+}5$ R@1).
The temperature $\tau$ is stable in $[0.05, 0.1]$; extreme values
($\leq 0.01$ or $\geq 0.5$) hurt because the softmax becomes either
too peaked or too flat for effective distillation.
The fusion weight $\omega$ is stable in $[5, 20]$; at $\omega{=}0$ (pure
T5 decoding) R@1 drops by 4 points, while $\omega{=}50$ over-weights
the geometric term and degrades performance.

\end{document}